\documentclass[
a4paper,
]{article}

\usepackage
[
]
{geometry}

\usepackage{setspace}
\usepackage{Settings}

\newcommand{\FBV}{\mathop{\bigvee}
}

\newcommand{\FBW}{\mathop{\bigwedge}
}

\usepackage{lineno}

\title{A minimal coalition logic}
\author{
Yinfeng Li${}^{1,2}$ and Fengkui Ju${}^{3,4}$\footnote{Corresponding author} \vspace{5pt} \\
{\small {$^1$IRIT-CNRS, University of Toulouse, France}} \\
{\small {$^2$\href{mailto:yinfeng.li@irit.fr}{yinfeng.li@irit.fr}}} \vspace{2.5pt} \\
{\small {$^3$School of Philosophy, Beijing Normal University, China}} \\
{\small {$^4$\href{mailto:fengkui.ju@bnu.edu.cn}{fengkui.ju@bnu.edu.cn}}}
}
\date{}

\begin{document}

\maketitle

\setlength{\parskip}{0.5em}


\begin{abstract}

\noindent Coalition Logic is an important logic in logical studies of strategic reasoning, whose models are concurrent game models.
In this paper, first, we systematically discuss three assumptions of concurrent game models and argue that they are too strong.
The first is seriality; that is, every coalition always has an available joint action.
The second is the independence of agents; that is, the merge of two available joint actions of two disjoint coalitions is always an available joint action of the union of the two coalitions.
The third is determinism; that is, all available joint actions of the grand coalition always have a unique outcome.
Second, we present a coalition logic based on general concurrent game models that do not have the three assumptions and show its completeness.
This logic seems minimal for reasoning about coalitional powers.

\medskip

\noindent \textbf{Keywords:} Concurrent game models, seriality, independence of agents, determinism, general concurrent game models

\end{abstract}

\section{Introduction}

\subsubsection*{Coalition Logic $\FCL$}

Coalition Logic $\FCL$ is a logic for reasoning about coalitional powers, proposed by Pauly \cite{pauly_logic_2001,pauly_modal_2002}.
The language of $\FCL$ is a modal language with the featured operator $\Fclo{\FAA} \phi$, indicating \emph{some available joint action of the coalition $\FAA$ ensures $\phi$}.
Its models are \emph{concurrent game models}. Roughly, in a concurrent game model:
there are some states;
there are some agents who form coalitions;
at any state, every coalition has some available joint actions;
every joint action of a coalition has some possible outcome states.
The formula $\Fclo{\FAA} \phi$ is true at a state in a concurrent game model if $\FAA$ has an available joint action such that $\phi$ is true at every possible outcome state of the action. 

Coalition Logic $\FCL$ is an important logic for strategic reasoning. Many logics in this field are extensions of $\FCL$. For example, Alternating-time Temporal Logic $\FATL$ (\cite{alur_alternating-time_2002}) is a temporal extension of $\FCL$, and strategy Logic $\FSL$ \cite{mogavero_reasoning_2014} is an extension of $\FCL$, whose language has quantifiers for and names of strategies. We refer to \cite{benthem_models_2015}, \cite{agotnes_knowledge_2015} and \cite{sep-logic-power-games} for overviews of the area.

\subsubsection*{Our work}

Concurrent game models rest on several assumptions. In this paper, we consider three of them.
The first one is \emph{seriality}: coalitions always have available joint actions. The second one is \emph{the independence of agents}: the merge of two available joint actions of two disjoint coalitions is always an available joint action of the union of the two coalitions. The third one is \emph{determinism}: available joint actions of the grand coalition always have a unique outcome.

In this work, we first systematically discuss the three assumptions and argue that they are too strong.
Second, present a Minimal Coalition Logic $\FMCL$ based on \emph{general concurrent game models} which do not have the three assumptions, compare $\FMCL$ to $\FCL$ in detail, and show its completeness.

The logic $\FMCL$ is \emph{minimal} in the following sense: general concurrent game models cannot be more ``general'', given that they are for coalitional powers.

\subsubsection*{Related work}

There has been some work in the agency literature that drops some of the three assumptions.

Jiang and Naumov \cite{JIANG2022103727} presented a logic for reasoning about knowledge and strategies in
multi-agent systems, whose framework is in the tradition of coalition logic. In their models, Jiang and Naumov allowed an action profile to have no outcome state or more than one outcome state. They also allowed dead ends, that is, states where no action profile has an outcome state
\footnote{In \cite{JIANG2022103727}, Jiang and Naumov did not explicitly define available joint actions.
There are two ways to understand this. One is that there is no such notion of available joint actions. Then it makes no sense to say their models are serial/deterministic or not, as seriality and determinism are about available joint actions.
Another way is that they implicitly treated all joint actions, including those without any outcome state, as available. Then, their models are non-deterministic but are serial.
}. They did this to model termination of games.

The independence of agents commonly holds in STIT (see to it that) logics. Sergot \cite{sergot_examples_2014} found no convincing justification for this principle in the literature of STIT logics. He offered examples showing that the principle fails.
Based on an extension of Propositional Dynamic Logic with multi-agents, Royakkers and Hughes \cite{royakkers_blame_2020} formalized three notions of responsibilities. They showed that the independence of agents does not hold generally and argued that it is also inappropriate for one notion of responsibility, that is, accountability. They use models where the independence of agents fails. In addition, determinism does not hold in their models either. They allowed this for the reason that agents may perform actions with unknown outcomes.

Boudou and Lorini \cite{10.5555/3237383.3237443} introduced new semantics for a temporal STIT logic, utilizing non-deterministic concurrent game models. 
Sergot \cite{sergot_actual_2022} investigated how causal responsibility can be treated in STIT logics, where he explicitly argued that there are non-deterministic scenarios.
To deal with outcome uncertainty of agents' actions, some work, such as \cite{chen_probabilistic_2007} and \cite{naumov_strategic_2021}, has introduced probability to concurrent game models. In these work, available joint actions of the grand coalition at a state may have different outcome states with different probabilities.

\subsubsection*{Our brief reasons to drop the three assumptions}

Many \emph{artificial settings} (such as games) have terminal states, where agents do not have available actions. To model terminal states of these artificial settings, we need to give up seriality.

Agents' actions often interact with each other, and whether an agent can perform an action is often conditional on other agents' actions at the same time. In these situations, the independence of agents fails. To model agents' conditional abilities, we need to drop the assumption of agent independence.

There are many situations where some joint actions of all behaving agents have more than one outcome state.
Additionally, in many situations, we just consider some but not all agents who form the grand coalition.
In these situations, determinism does not hold.

\subsubsection*{Structures of the paper}

The rest of the paper is structured as follows.
In Section \ref{section:General settings of coalition logics}, we present some general settings of coalition logics.
In Section \ref{section:Coalition Logic CL}, we give $\FCL$ in detail.
In Section \ref{section:Three assumptions in concurrent game models are too strong}, we discuss the three assumptions and show that they do not hold in all situations.
In Section \ref{section:A minimal coalition logic MCL based on general concurrent game models}, we present $\FMCL$, argue that $\FMCL$ is a minimal coalition logic for reasoning about coalitional powers, and discuss whether $\FCL$ can simulate $\FMCL$. 
We show the completeness of $\FMCL$ in Section \ref{section:Completeness of MCL}.
We point out some further work in Section \ref{section:Further work}.

\section{General settings of coalition logics}
\label{section:General settings of coalition logics}

In this section, we introduce some general settings of coalition logics, including their language, models, and semantics. These general settings will ease things later.

\subsection{Language}

Let $\FAG$ be a nonempty finite set of \emph{agents} and $\FAP$ be a countable set of atomic propositions. Each (possibly empty) subset $\FAA$ of $\FAG$ is called a \Fdefs{coalition}. $\FAG$ is called the \Fdefs{grand coalition}.
In the sequel, for any $a \in \FAG$, we will often write $a$ instead of $\{a\}$, given that no confusion occurs.

\begin{definition}[The language $\Phi$]
\label{definition:The language Phi}
The language $\Phi$ is defined as follows, where $p$ ranges over $\FAP$ and $\FAA \subseteq \FAG$:
\[
\phi ::=\top \mid p \mid \neg \phi \mid (\phi \wedge \phi) \mid \Fclo{\FAA} \phi
\]
\end{definition}

The formula $\Fclo{\FAA} \phi$ indicates that \emph{some available joint action of $\FAA$ ensures $\phi$}. Here, the notation of $\Fclo{\cdot} \phi$ differs from its notation in the literature. We do this to indicate the iteration of quantifiers in its meaning, which will be seen below.

Here are some derivative expressions:
\begin{itemize}

\item 

The propositional connectives $\bot, \lor, \rightarrow$, and $\leftrightarrow$ are defined as usual.

\item 

Define the dual $\Fclod{\FAA} \phi$ of $\Fclo{\FAA} \phi$ as $\neg \Fclo{\FAA} \neg \phi$, indicating that \emph{every available joint action of $\FAA$ enables $\phi$}.

\item 

$\Box \phi$, defined as $\Fclo{\emptyset} \top \rightarrow \Fclo{\emptyset} \phi$, means that \emph{$\phi$ will necessarily be true}.

\item 

$\Diamond \phi$, defined as $\Fclo{\emptyset} \top \land \Fclod{\emptyset} \phi$, means that \emph{$\phi$ will possibly be true}.

\end{itemize}

This language is the language of $\FCL$ and $\FMCL$.

\subsection{Abstract multi-agent action models}

Let $\FAC$ be a nonempty set of \emph{actions}.
\begin{itemize}

\item

For every coalition $\FAA$, a function $\ja{\FAA}:\FAA \to \FAC$ is called a \Fdefs{joint action} of $\FAA$. Note, especially, that $\emptyset$ is the only joint action of the empty coalition. A joint action of $\FAG$ is called an \Fdefs{action profile}.

For every $\FAA \subseteq \FAG$, we define $\FJA_\FAA = \{\sigma_\FAA \mid \sigma_\FAA: \FAA \rightarrow \FAC\}$, which is the set of joint actions of $\FAA$.
Note $\FJA_\emptyset = \{\emptyset\}$.
%

As joint actions are treated as unary functions, which are sets of binary tuples, set-theoretical notions such as union and intersection apply to them.

In the sequel, we sometimes use sequences of actions to indicate joint actions of coalitions, where an implicit order among agents is supposed.

\item

Let $\FAA$ be a coalition, and $\FBB \subseteq \FAA$. Let $\sigma_\FAA$ be a joint action of $\FAA$. We use $\sigma_\FAA |_\FBB$ to denote the subset of $\sigma_\FAA$ that is a joint action of $\FBB$, called the \Fdefs{restriction} of $\sigma_\FAA$ to $\FBB$. Respectively, $\sigma_\FAA$ is called an \Fdefs{extension} of $\sigma_\FAA|_\FBB$.
Given a set of joint actions $\Sigma_\FAA$ of $\FAA$, we define $\Sigma_\FAA |_\FBB := \{\sigma_\FAA|_\FBB \mid \sigma_\FAA \in \Sigma_\FAA\}$.

\end{itemize}

\begin{definition}[Abstract multi-agent action models]
\label{definition:Concurrent game models}
An \Fdefs{abstract multi-agent action model} is a tuple $\MM = (\FST, \FAC, \{\Fav_\FAA \mid \FAA \subseteq \FAG\}, \{\Fout_\FAA \mid \FAA \subseteq \FAG\}, \Flab)$, where:
\begin{itemize}

\item

$\FST$ is a nonempty set of states.

\item

$\FAC$ is a nonempty set of actions.

\item

for every $\FAA \subseteq \FAG$, $\Fav_\FAA: \FST \rightarrow \mathcal{P}(\FJA_\FAA)$ is an \Fdefs{availability function} for $\FAA$.

\emph{Here, $\Fav_\FAA (s)$ is the set of all available joint actions of $\FAA$ at $s$.}

\item

for every $\FAA \subseteq \FAG$, $\Fout_\FAA: \FST \times \FJA_\FAA \rightarrow \mathcal{P}(\FST)$ is an \Fdefs{outcome function} for $\FAA$.

\emph{Here, $\Fout_\FAA (s, \sigma_\FAA)$ is the set of outcome states of $\FAA$ performing $\sigma_\FAA$ at $s$.}

\item

$\Flab: \FST \rightarrow \mathcal{P}(\FAP)$ is a \Fdefs{labeling function}.

\end{itemize}

\end{definition}

To be used to represent real scenarios, abstract multi-agent action models should satisfy some constraints, which is why they are called \emph{abstract}.
Later, models of $\FCL$ and $\FMCL$ will be defined as abstract multi-agent action models meeting certain conditions.

\paragraph{Remarks}

As we can see, the language of coalition logics does not talk about unavailable joint actions. As a result, how unavailable joint actions are treated does not matter technically.

We understand outcome states as follows: a state $t$ is an outcome state of a joint action $\sigma_\FAA$ of a coalition $\FAA$ at a state $s$ if and only if $t$ is a \emph{possible next-moment state} where $\sigma_\FAA$ has just been done by $\FAA$.

We understand available joint actions as follows: a joint action $\sigma_\FAA$ of a coalition $\FAA$ is available at a state $s$ if and only if $\sigma_\FAA$ has an outcome state at $s$.
We say that a joint action of a coalition is \Fdefs{conditionally available} if its availability depends on the simultaneous actions of other agents.

Note that availability is in an \emph{ontic} sense: it might be the case that some joint action is available for a coalition, but the coalition does not know it, or the coalition does not want to perform it, or it is illegal.

\subsection{Semantics}

\begin{definition}[Semantics of $\Phi$]
\label{definition:Semantics of Phi}
~

\begin{tabular}{lll}
$\MM, s \Vdash \top$ & & \\
$\MM, s \Vdash p$ & $\Leftrightarrow$ & \parbox[t]{28em}{$p \in \Flab (s)$} \\
$\MM, s \Vdash \neg \phi$ & $\Leftrightarrow$ & \parbox[t]{28em}{not $\MM, s \Vdash \phi$} \\
$\MM, s \Vdash \phi \land \psi$ & $\Leftrightarrow$ & \parbox[t]{28em}{$\MM, s \Vdash \phi$ and $\MM, s \Vdash \psi$} \\
$\MM, s \Vdash \Fclo{\FAA} \phi$ & $\Leftrightarrow$ & \parbox[t]{28em}{there is $\sigma_\FAA \in \Fav_\FAA (s)$ such that for all $t \in \Fout_\FAA (s, \sigma_\FAA)$, $\MM, t \Vdash \phi$}
\end{tabular}

\end{definition}

It can be verified:

\medskip

\begin{tabular}{lll}
$\MM, s \Vdash \Fclod{\FAA} \phi$ & $\Leftrightarrow$ & \parbox[t]{28em}{for all $\sigma_\FAA \in \Fav_\FAA (s)$, there is $t \in \Fout_\FAA (s, \sigma_\FAA)$ such that $\MM, t \Vdash \phi$} \\
$\MM, s \Vdash \Box \phi$ & $\Leftrightarrow$ & \parbox[t]{28em}{if $\emptyset \in \Fav_\emptyset (s)$, then for all $t \in \Fout_\emptyset (s, \emptyset)$, $\MM, t \Vdash \phi$} \\
$\MM, s \Vdash \Diamond \phi$ & $\Leftrightarrow$ & \parbox[t]{28em}{$\emptyset \in \Fav_\emptyset (s)$ and there is $t \in \Fout_\emptyset (s, \emptyset)$ such that $\MM, t \Vdash \phi$}
\end{tabular}

\medskip

This semantics is the semantics of $\FCL$ and $\FMCL$.

Let $(\MM,s)$ be a pointed abstract multi-agent action model, $\sigma_\FAA$ be an available joint action of a coalition $\FAA$ at $s$, and $\phi$ be a formula. We use $\sigma_\FAA \leadto_{(\MM,s)} \phi$ to indicate that $\sigma_\FAA$ \emph{ensures} $\phi$ at $(\MM,s)$, that is, for all $t \in \Fout_\FAA (s,\sigma_\FAA), \MM,s \Vdash \phi$.

\section{Coalition Logic $\FCL$}
\label{section:Coalition Logic CL}

In this section, we present models and axiomatization of $\FCL$ and make some discussion of them.

\subsection{Concurrent game models}

The definition of concurrent game models given below is slightly different from but equivalent to the definition given in \cite{pauly_modal_2002}.

We first introduce some auxiliary notions and notation.

Let $\FAA$ and $\FBB$ be two disjoint coalitions, $\Sigma_\FAA$ be a set of joint actions of $\FAA$, and $\Sigma_\FBB$ be a set of joint actions of $\FBB$. We define $\Sigma_\FAA \oplus \Sigma_\FBB$ as $\{\sigma_\FAA \cup \sigma_\FBB \mid \sigma_\FAA \in \Sigma_\FAA \mbox { and } \sigma_\FBB \in \Sigma_\FBB\}$, which is a set of joint actions of $\FAA \cup \FBB$.

Let $\{\FAA_i \mid i \in I\}$ be a family of pairwise disjoint coalitions for some (possibly empty) index set $I$, and $\Lambda = \{\Sigma_{\FAA_i} \mid i \in I \text{ and } \Sigma_{\FAA_i} \text{is a set of joint actions of } \FAA_i \}$.

A function $f: I \rightarrow \bigcup \Lambda$ is a \Fdefs{choice function} of $\Lambda$ if for every $i \in I$, $f(i) \in \Sigma_{\FAA_i}$. Note that if $I = \emptyset$, then $f = \emptyset$.

Define:
\[
\bigoplus \Lambda = \{\sigma: \bigcup_{i \in I} \FAA_i \to \FAC \mid \text{there is a choice function $f$ of } \Lambda \text{ such that } \sigma = \bigcup_{i \in I} f(i) \},
\]
which is a set of joint actions of $\bigcup_{i \in I}\FAA_i$.

Note $\bigoplus \emptyset = \{\emptyset\}$. 

Here is an example of $\bigoplus$.
Let $\Sigma_a = \{\sigma_a^1, \sigma_a^2\}$, $\Sigma_b = \{\sigma_b^1, \sigma_b^2\}$ and $\Sigma_c = \{\sigma_c^1\}$. Then $\bigoplus \{\Sigma_a, \Sigma_b, \Sigma_c\}$ consists of the following elements:
\begin{itemize}

\item 

$\sigma_{\{a,b,c\}}^1 = \bigcup \{\sigma_a^1, \sigma_b^1, \sigma_c^1\} = \sigma_a^1 \cup \sigma_b^1 \cup \sigma_c^1$

\item 

$\sigma_{\{a,b,c\}}^2 = \bigcup \{\sigma_a^1, \sigma_b^2, \sigma_c^1\} = \sigma_a^1 \cup \sigma_b^2 \cup \sigma_c^1$

\item 

$\sigma_{\{a,b,c\}}^3 = \bigcup \{\sigma_a^2, \sigma_b^1, \sigma_c^1\} = \sigma_a^2 \cup \sigma_b^1 \cup \sigma_c^1$

\item 

$\sigma_{\{a,b,c\}}^4 = \bigcup \{\sigma_a^2, \sigma_b^2, \sigma_c^1\} = \sigma_a^2 \cup \sigma_b^2 \cup \sigma_c^1$

\end{itemize}

\begin{definition}[Concurrent game models]
\label{definition:Concurrent game models}
An abstract multi-agent action model $\MM = (\FST, \FAC,$ $\{\Fav_\FAA \mid \FAA \subseteq \FAG\}, \{\Fout_\FAA \mid \FAA \subseteq \FAG\}, \Flab)$ is a \Fdefs{concurrent game model} if:
\begin{enumerate}[label=(\arabic*),leftmargin=3.33em]

\item 

for every $a \in \FAG$ and $s \in \FST$, $\Fav_a (s)$ is not empty.

\emph{This constraint indicates that concurrent game models are serial: every coalition always has an available joint action.}

\item 

for every $\FAA \subseteq \FAG$ and $s \in \FST$, $\Fav_\FAA (s) = \bigoplus \{\Fav_{a} (s) \mid a \in \FAA\}$.

\emph{Note $\Fav_\emptyset (s) = \{\emptyset\}$. Intuitively, $\Fav_\FAA$ is determined by all $\Fav_{a}$, where $a$ is in $\FAA$.
}

\item 

for every $s \in \FST$ and $\sigma_\FAG \in \FJA_\FAG$, if $\sigma_\FAG \in \Fav_\FAG (s)$, then $\Fout_\FAG (s, \sigma_\FAG)$ is a singleton, otherwise $\Fout_\FAG (s, \sigma_\FAG) = \emptyset$.

\emph{This constraint indicates that concurrent game models are deterministic: every available joint action of the grand coalition has a unique outcome.}

\item

for every $\FAA \subseteq \FAG$, $s \in \FST$ and $\sigma_\FAA \in \FJA_\FAA$, $\Fout_\FAA (s, \sigma_\FAA) = \bigcup \{\Fout_\FAG (s, \sigma_\FAG) \mid \sigma_\FAG \in \FJA_\FAG \text{ and}$ $\sigma_\FAA \subseteq \sigma_\FAG\}$.

\emph{Intuitively, $\Fout_\FAA$ is determined by $\Fout_\FAG$.}
\emph{Note $\Fout_\emptyset (s,\emptyset) = \bigcup \{\Fout_\FAG (s, \sigma_\FAG) \mid \sigma_\FAG \in \FJA_\FAG\}$, which is the set of all successors of $s$.}

\end{enumerate}

\end{definition}


The following fact indicates that availability functions in concurrent game models coincide with our understanding of available joint actions given above.

\begin{fact}
Let $\MM = (\FST, \FAC, \{\Fav_\FAA \mid \FAA \subseteq \FAG\}, \{\Fout_\FAA \mid \FAA \subseteq \FAG\}, \Flab)$ be a concurrent game model.
Then, for every $\FAA \subseteq \FAG$ and $s \in \FST$, $\Fav_\FAA (s) = \{\sigma_\FAA \in \FJA_\FAA \mid \Fout_\FAA (s, \sigma_\FAA) \neq \emptyset\}$.
\end{fact}

\begin{proof}
~

Let $\FAA \subseteq \FAG$ and $s \in \FST$.

Assume $\FAA = \emptyset$.

Note $\Fav_\emptyset (s) = \{\emptyset\} = \FJA_\emptyset$. It suffices to show $\Fout_\emptyset (s,\emptyset) \neq \emptyset$. Note $\Fav_\FAG (s) \neq \emptyset$. Let $\sigma_\FAG \in \Fav_\FAG (s)$. Then $\Fout_\FAG (s, \sigma_\FAG) \neq \emptyset$.
Note $\Fout_\emptyset (s,\emptyset) = \bigcup \{\Fout_\FAG (s, \sigma_\FAG) \mid \sigma_\FAG \in \FJA_\FAG\}$. Then $\Fout_\emptyset (s,\emptyset) \neq \emptyset$.

Assume $\FAA = \{a_1, \dots, a_n\}$.

Let $\sigma_\FAA \in \Fav_\FAA (s)$. Then there is $\sigma_{a_1} \in \Fav_{a_1} (s)$, \dots, $\sigma_{a_n} \in \Fav_{a_n} (s)$ such that $\sigma_\FAA = \sigma_{a_1} \cup \dots \cup \sigma_{a_n}$.
Assume $\FAA = \FAG$. Then $\Fout_\FAA (s,\sigma_\FAA) \neq \emptyset$.
Assume $\FAA \neq \FAG$. Let $\FAG - \FAA = \{b_1, \dots, b_m\}$. Let $\sigma_{b_1} \in \Fav_{b_1} (s)$, \dots, $\sigma_{b_m} \in \Fav_{b_m} (s)$. Let $\sigma_\FAG = \sigma_{a_1} \cup \dots \cup \sigma_{a_n} \cup \sigma_{b_1} \cup \dots \cup \sigma_{b_m}$. Then $\sigma_\FAG \in \Fav_\FAG (s)$. Then $\Fout_\FAG (s,\sigma_\FAG) \neq \emptyset$.
Note $\sigma_\FAA \subseteq \sigma_\FAG$. Then $\Fout_\FAA (s, \sigma_\FAA) \neq \emptyset$.

Assume $\sigma_\FAA \in \{\sigma_\FAA \in \FJA_\FAA \mid \Fout_\FAA (s, \sigma_\FAA) \neq \emptyset\}$. Then $\Fout_\FAA (s, \sigma_\FAA) \neq \emptyset$. Then there is $\sigma_\FAG \in \FJA_\FAG$ such that $\sigma_\FAA \subseteq \sigma_\FAG$ and $\Fout_\FAG (s, \sigma_\FAG) \neq \emptyset$. Then $\sigma_\FAG \in \Fav_\FAG (s)$.
If $\FAA = \FAG$, we are done.
Assume $\FAA \neq \FAG$. Let $\FAG - \FAA = \{b_1, \dots, b_m\}$. Then there is $\sigma_{a_1} \in \Fav_{a_1} (s), \dots, \sigma_{a_n} \in \Fav_{a_n} (s), \sigma_{b_1} \in \Fav_{b_1} (s), \dots, \sigma_{b_m} \in \Fav_{b_m} (s)$ such that $\sigma_\FAG = \sigma_{a_1} \cup \dots \cup \sigma_{a_n} \cup \sigma_{b_1} \cup \dots \cup \sigma_{b_m}$. Then $\sigma_\FAA = \sigma_{a_1} \cup \dots \cup \sigma_{a_n}$. Then $\sigma_\FAA \in \Fav_\FAA (s)$.

\end{proof}

The following fact gives two equivalent sets of constraints on availability functions in concurrent game models, from which we can see that available joint actions in concurrent game models are \emph{unconditional}.

\begin{fact}
\label{fact:Characterization of available joint action functions of concurrent game models}
Let $\MM = (\FST, \FAC, \{\Fav_\FAA \mid \FAA \subseteq \FAG\}, \{\Fout_\FAA \mid \FAA \subseteq \FAG\}, \Flab)$ be an abstract multi-agent action model.
Then, the following three conditions are equivalent:
\begin{enumerate}[label=(\arabic*),leftmargin=3.33em]

\item

for every $a \in \FAG$, $\FAA \subseteq \FAG$ and $s \in \FST$:
\begin{enumerate}

\item 

$\Fav_a (s)$ is not empty;

\item 

$\Fav_\FAA (s) = \bigoplus \{\Fav_a (s) \mid a \in \FAA\}$.

\end{enumerate}

\item

for every $\FAA, \FBB \subseteq \FAG$ such that $\FAA \cap \FBB = \emptyset$, and $s \in \FST$:
\begin{enumerate}

\item 

$\Fav_\FAA (s)$ is nonempty;

\item 

for every $\sigma_{\FAA \cup \FBB} \in \Fav_{\FAA \cup \FBB} (s)$, $\sigma_{\FAA \cup \FBB}|_\FAA \in \Fav_\FAA (s)$;

\item 

for every $\sigma_\FAA \in \Fav_\FAA (s)$ and $\sigma_\FBB \in \Fav_\FBB (s)$, $\sigma_\FAA \cup \sigma_\FBB \in \Fav_{\FAA \cup \FBB} (s)$.

\emph{This condition indicates the independence of agents: the merge of two available joint actions of two disjoint coalitions is always an available joint action of the union of the two coalitions.
}

\end{enumerate}

\item

For every $\FAA \subseteq \FAG$ and $s \in \FST$:
\begin{enumerate}

\item 

$\Fav_\FAA (s)$ is nonempty;

\item 

for every $\sigma_\FAG \in \Fav_\FAG (s)$, $\sigma_\FAG|_\FAA \in \Fav_\FAA (s)$;

\item 

for every $\sigma_\FAA \in \Fav_\FAA (s)$ and $\sigma_\FAAb \in \Fav_{\FAAb} (s)$, $\sigma_\FAA \cup \sigma_\FAAb \in \Fav_\FAG (s)$.

\end{enumerate}

\end{enumerate}

\end{fact}

\begin{proof}
~

(1) $\Rightarrow$ (2)

Assume (1). Let $\FAA, \FBB \subseteq \FAG$ such that $\FAA \cap \FBB = \emptyset$, and $s \in \FST$.

First, we show (2a).
Assume $\FAA = \emptyset$. As mentioned, $\Fav_\emptyset (s) = \{\emptyset\}$.
Assume $\FAA \neq \emptyset$. Let $\FAA = \{a_1, \dots, a_n\}$. Let $\sigma_{a_1} \in \Fav_{a_1} (s)$, \dots, $\sigma_{a_n} \in \Fav_{a_n} (s)$. Let $\sigma_\FAA = \sigma_{a_1} \cup \dots \cup \sigma_{a_n}$. Then $\sigma_\FAA \in \Fav_\FAA (s)$.

Second, we show (2b). 
Let $\FAja{\FAA \cup \FBB} \in \Fav_{\FAA \cup \FBB} (s)$. We want to show $\sigma_{\FAA \cup \FBB}|_\FAA \in \Fav_\FAA (s)$.

Assume $\FAA = \emptyset$. It is easy to check $\sigma_{\FAA \cup \FBB}|_\FAA \in \Fav_\FAA (s)$.

Assume $\FBB = \emptyset$. It is easy to check $\sigma_{\FAA \cup \FBB}|_\FAA \in \Fav_\FAA (s)$.

Assume $\FAA \neq \emptyset$ and $\FBB \neq \emptyset$. Let $\FAA = \{a_1, \dots, a_n\}$ and $\FBB = \{b_1, \dots, b_m\}$.
By (1b), there is $\sigma_{a_1} \in \Fav_{a_1} (s)$, \dots, $\sigma_{a_n} \in \Fav_{a_n} (s)$, $\sigma_{b_1} \in \Fav_{b_1} (s)$, \dots, $\sigma_{b_m} \in \Fav_{b_m} (s)$ such that $\sigma_{\FAA \cup \FBB} = \sigma_{a_1} \cup \dots \cup \sigma_{a_n} \cup \sigma_{b_1} \cup \dots \cup \sigma_{b_m}$. Then $\sigma_{\FAA \cup \FBB}|_\FAA = \sigma_{a_1} \cup \dots \cup \sigma_{a_n}$. By (1b), $\sigma_{\FAA \cup \FBB}|_\FAA \in \Fav_\FAA (s)$.

Third, we show (2c). Let $\sigma_\FAA \in \Fav_\FAA (s)$ and $\sigma_{\FBB} \in \Fav_\FBB (s)$. We want to show $\FAja{\FAA} \cup \FAja{\FBB} \in \Fav_{\FAA \cup \FBB} (s)$.

Assume $\FAA = \emptyset$. It is easy to check $\FAja{\FAA} \cup \FAja{\FBB} \in \Fav_{\FAA \cup \FBB} (s)$.

Assume $\FBB = \emptyset$. It is easy to check $\FAja{\FAA} \cup \FAja{\FBB} \in \Fav_{\FAA \cup \FBB} (s)$.

Assume $\FAA \neq \emptyset$ and $\FBB \neq \emptyset$. Let $\FAA = \{a_1, \dots, a_n\}$ and $\FBB = \{b_1, \dots, b_m\}$.
By (1b), $\sigma_\FAA|_{a_1} \in \Fav_{a_1} (s)$, \dots, $\sigma_\FAA|_{a_n} \in \Fav_{a_n} (s)$, $\sigma_{\FBB}|_{b_1} \in \Fav_{b_1} (s)$, \dots, $\sigma_{\FBB}|_{b_m} \in \Fav_{b_m} (s)$. By (1b), $\sigma_\FAA|_{a_1} \cup \dots \cup \sigma_\FAA|_{a_n} \cup \sigma_{\FBB}|_{b_1} \cup \dots \cup \sigma_{\FBB}|_{b_m} \in \Fav_{\FAA \cup \FBB} (s)$, that is, $\FAja{\FAA} \cup \FAja{\FBB} \in \Fav_{\FAA \cup \FBB} (s)$.

\medskip

(2) $\Rightarrow$ (3)

It is easy to see that (3) is a special case of (2).

\medskip

(3) $\Rightarrow$ (1)

Assume (3). Let $\FAA \subseteq \FAG$ and $s \in \FST$. It suffices to show (1b).

Assume $\FAA = \emptyset$. As mentioned, $\Fav_\emptyset (s) = \{\emptyset\} = \bigoplus \{\Fav_a (s) \mid a \in \FAA\}$.

Assume $\FAA \neq \emptyset$. Let $\FAA = \{a_1, \dots, a_n\}$.

Let $\sigma_\FAA \in \Fav_\FAA (s)$. We want to show $\FAja{\FAA} \in \bigoplus \{\Fav_a (s) \mid a \in \FAA\}$.
By (3a), $\Fav_{\FAAb} (s)$ is not empty. Let $\sigma_\FAAb \in \Fav_{\FAAb} (s)$. By (3c), $\FAja{\FAA} \cup \FAja{\FAAb} \in \Fav_\FAG (s)$. Let $\FAja{\FAG} = \FAja{\FAA} \cup \FAja{\FAAb}$. By (3b), $\sigma_\FAG|_{a_1} \in \Fav_{a_1} (s), \dots, \sigma_\FAG|_{a_n} \in \Fav_{a_n} (s)$. Then $\sigma_\FAG|_{a_1} \cup \dots \cup \sigma_\FAG|_{a_n} \in \bigoplus \{\Fav_a (s) \mid a \in \FAA\}$, that is, $\FAja{\FAA} \in \bigoplus \{\Fav_a (s) \mid a \in \FAA\}$.

Let $\FAja{\FAA} \in \bigoplus \{\Fav_a (s) \mid a \in \FAA\}$. Then there is $\sigma_{a_1} \in \Fav_{a_1} (s), \dots, \sigma_{a_n} \in \Fav_{a_n} (s)$ such that $\sigma_\FAA = \sigma_{a_1} \cup \dots \cup \sigma_{a_n}$. We want to show $\sigma_\FAA \in \FAajas{\FAA}$.

By (3a), $\FAajas{\FAG}$ is not empty. Let $\delta_{\FAG} \in \FAajas{\FAG}$. For any $\FBB \subseteq \FAG$, we use $\delta_{\FBB}$ to indicate $\delta_{\FAG}|_\FBB$.

By (3b), $\delta_{\FAG - \{a_1\}} \in \FAajas{\FAG - \{a_1\}}$. By (3c), $\sigma_{a_1} \cup \delta_{\FAG - \{a_1\}} \in \FAajas{\FAG}$.

By (3b), $\sigma_{a_1} \cup \delta_{\FAG - \{a_1,a_2\}} \in \FAajas{\FAG -\{a_2\}}$. By (3c), $\sigma_{a_1} \cup \sigma_{a_2} \cup \delta_{\FAG - \{a_1,a_2\}} \in \FAajas{\FAG}$.

\dots.

Finally, we know $\sigma_{a_1} \cup \dots \cup \sigma_{a_n} \cup \delta_{\FAG - \{a_1, \dots, a_n\}} \in \FAajas{\FAG}$.

By (3b), $\sigma_{a_1} \cup \dots \cup \sigma_{a_n} \in \FAajas{\FAA}$, that is, $\sigma_\FAA \in \FAajas{\FAA}$.

\end{proof}


What follows is an example of how scenarios are represented by pointed concurrent game models.

\begin{example}
\label{example:two masks}

Adam and Bob work in a closed gas laboratory when an earthquake occurs. They find that the door cannot be opened.

One toxic gas container is damaged and will soon leak. Fortunately, there are two gas masks in the laboratory.

Adam and Bob have two available actions in this situation: \emph{doing nothing} and \emph{wearing a gas mask}. If they both do nothing, they will both die. If they both wear a gas mask, they will both survive. If Adam does nothing and Bob wears a gas mask, Adam will die and Bob will survive. If Adam wears a gas mask and Bob does nothing, Adam will survive and Bob will die.

We suppose that \emph{one can only do nothing after wearing a gas mask}. This scenario can be represented by the pointed concurrent game model depicted in Figure \ref{figure:two masks}.

\end{example}

\begin{figure}
\begin{center}
\begin{tikzpicture}
[
->=stealth,
scale=1,
every node/.style={transform shape},
]

\tikzstyle{every state}=[minimum size=10mm]

\node[state,fill=gray!15] (s-0) {$s_0$};
\node[state,position=0:{25mm} from s-0] (s-1) {$s_1$};
\node[state,position=45:{25mm} from s-0] (s-2) {$s_2$};
\node[state,position=135:{25mm} from s-0] (s-3) {$s_3$};
\node[state,position=180:{25mm} from s-0] (s-4) {$s_4$};

\node[below=5mm] (a-s-0) at (s-0) {$\{l_a,l_b\}$};
\node[above=5mm] (a-s-1) at (s-1) {$\{m_a,l_a\}$};
\node[above=5mm] (a-s-2) at (s-2) {$\{m_b,l_b\}$};
\node[above=5mm] (a-s-3) at (s-3) {$\{\}$};
\node[above=5mm] (a-s-4) at (s-4) {$\{m_a,m_b,l_a,l_b\}$};

\path
(s-0) edge [above] node {$\mathtt{(w,n)}$} (s-1)
(s-0) edge [left] node {$\mathtt{(n,w)}$} (s-2)
(s-0) edge [right] node {$\mathtt{(n,n)}$} (s-3)
(s-0) edge [above] node {$\mathtt{(w,w)}$} (s-4)
;

\path
(s-1) edge [loop right,midway] node {$\mathtt{(n,n)}$} (s-1)
(s-2) edge [loop right,midway] node {$\mathtt{(n,n)}$} (s-2)
(s-3) edge [loop left,midway] node {$\mathtt{(n,n)}$} (s-3)
(s-4) edge [loop left,midway] node {$\mathtt{(n,n)}$} (s-4)
;

\end{tikzpicture}

\caption{
This figure indicates a pointed concurrent game model $(\MM,s_0)$ for the situation of Example \ref{example:two masks}.
Here: $m_a, m_b, l_a$ and $l_b$ respectively express the propositions \emph{Adam is wearing a gas mask}, \emph{Bob is wearing a gas mask}, \emph{Adam is alive} and \emph{Bob is alive}; $\mathtt{n}$ and $\mathtt{w}$ respectively represent the actions \emph{doing nothing} and \emph{wearing a mask}; an arrow from a state $x$ to a state $y$ labeled with an action profile $\alpha$, indicates that $y$ is a possible outcome state of performing $\alpha$ at $x$.
}

\label{figure:two masks}

\end{center}
\end{figure}

\subsection{Language and semantics}

In $\FCL$, the formula $\Fclo{\FAA} \phi$ intuitively indicates \emph{some unconditionally available joint action of $\FAA$ ensures $\phi$}, and $\Fclod{\FAA} \phi$ intuitively indicates \emph{every unconditionally available joint action of $\FAA$ enables $\phi$}.

Note that $\Box \phi$ is defined as $\Fclo{\emptyset} \top \rightarrow \Fclo{\emptyset} \phi$ and $\Diamond \phi$ is defined as $\Fclo{\emptyset} \top \land \Fclod{\emptyset} \phi$.
In concurrent game models, the empty action is always available for the empty coalition. Then, the following holds:

\medskip

\begin{tabular}{lll}
$\MM, s \Vdash \Box \phi$ & $\Leftrightarrow$ & \parbox[t]{28em}{for all $t \in \Fout_\emptyset (s, \emptyset)$, $\MM, t \Vdash \phi$} \\
$\MM, s \Vdash \Diamond \phi$ & $\Leftrightarrow$ & \parbox[t]{28em}{there is $t \in \Fout_\emptyset (s, \emptyset)$ such that $\MM, t \Vdash \phi$}
\end{tabular}

\medskip

We use $\models_\FCL \phi$ to indicate that $\phi$ is $\FCL$-\emph{valid}: $\phi$ is true at all pointed concurrent game models. We often drop ``$\FCL$'' when the contexts are clear.

\subsection{Axiomatization}

\begin{definition}[An axiomatic system for $\FCL$~\cite{pauly_modal_2002}]
\label{definition:An axiomatic system for CL}
~

\noindent \emph{Axioms}\footnote{In the literature such as \cite{Pacuit2017NeighborhoodSF}, the axioms $\mathtt{A}\text{-}\mathtt{Ser}$ and $\mathtt{A}\text{-}\mathtt{IA}$ are respectively called the axioms of \emph{safety} and \emph{superadditivity}.}:

\medskip

\begin{tabular}{rl}
Tautologies ($\mathtt{A}\text{-}\mathtt{Tau}$): & all propositional tautologies \vspace{5pt} \\
Monotonicity ($\mathtt{A}\text{-}\mathtt{Mon}$): & $\Fclo{\FAA} (\phi \land \psi) \rightarrow \Fclo{\FAA} \phi$ \vspace{5pt} \\
Liveness ($\mathtt{A}\text{-}\mathtt{Live}$): & $\neg \Fclo{\FAA} \bot$ \vspace{5pt} \\
Seriality ($\mathtt{A}\text{-}\mathtt{Ser}$): & $\Fclo{\FAA} \top$ \vspace{5pt} \\
Independence of agents ($\mathtt{A}\text{-}\mathtt{IA}$): & $(\Fclo{\FAA} \phi \land \Fclo{\FBB} \psi) \rightarrow \Fclo{\FAA \cup \FBB} (\phi \land \psi)$, where $\FAA \cap \FBB = \emptyset$ \vspace{5pt} \\
$\FAG$-maximality ($\mathtt{A}\text{-}\mathtt{Max}$): & $\neg \Fclo{\emptyset} \neg \phi \rightarrow \Fclo{\FAG} \phi$
\end{tabular}

\medskip

\noindent \emph{Inference rules}:

\medskip

\begin{tabular}{rl}
Modus ponens ($\mathtt{R}\text{-}\mathtt{MP}$): & $\dfrac{\phi, \phi \rightarrow \psi}
{\psi}$ \vspace{5pt} \\
Replacement of equivalence ($\mathtt{R}\text{-}\mathtt{RE}$): & $\dfrac{\phi \leftrightarrow \psi}
{\Fclo{\FAA} \phi \leftrightarrow \Fclo{\FAA} \psi}$
\end{tabular}

\end{definition}

We use $\vdash_\FCL \phi$ to indicate that $\phi$ is \Fdefs{derivable} in this system. We often drop ``$\FCL$'' when the contexts are clear.

\begin{theorem}[Soundness and completeness of $\FCL$~\cite{pauly_modal_2002,goranko_strategic_2013}]

The axiomatic system given in Definition \ref{definition:An axiomatic system for CL} is sound and complete with respect to the set of $\FCL$-valid formulas of $\Phi$.

\end{theorem}

To better compare $\FMCL$ to $\FCL$, we give an equivalent axiomatic system for $\FCL$. 

\begin{definition}[Another axiomatic system for $\FCL$]
\label{definition:Another axiomatic system for CL}
~

\noindent \emph{Axioms}:

\medskip

\begin{tabular}{rl}
Tautogies ($\mathtt{A}\text{-}\mathtt{Tau}$): & all propositional tautologies \vspace{5pt} \\
Monotonicity of goals ($\mathtt{A}\text{-}\mathtt{MG}$): & $\Fclo{\emptyset} (\phi \rightarrow \psi) \rightarrow (\Fclo{\FAA} \phi \rightarrow \Fclo{\FAA} \psi)$ \vspace{5pt} \\
Monotonicity of coalitions ($\mathtt{A}\text{-}\mathtt{MC}$): & $\Fclo{\FAA} \phi \rightarrow \Fclo{\FBB} \phi$, where $\FAA \subseteq \FBB$ \vspace{5pt} \\
Liveness ($\mathtt{A}\text{-}\mathtt{Live}$): & $\neg \Fclo{\FAA} \bot$ \vspace{5pt} \\
Seriality ($\mathtt{A}\text{-}\mathtt{Ser}$): & $\Fclo{\FAA} \top$ \vspace{5pt} \\
Independence of agents ($\mathtt{A}\text{-}\mathtt{IA}$): & $(\Fclo{\FAA} \phi \land \Fclo{\FBB} \psi) \rightarrow \Fclo{\FAA \cup \FBB} (\phi \land \psi)$, where $\FAA \cap \FBB = \emptyset$ \vspace{5pt} \\
Determinism ($\mathtt{A}\text{-}\mathtt{Det}$): & $\Fclo{\FAA} (\phi \lor \psi) \rightarrow (\Fclo{\FAA} \phi \lor \Fclo{\FAG} \psi)$
\end{tabular}

\medskip

\noindent \emph{Inference rules}:

\medskip

\begin{tabular}{rl}
Modus ponens ($\mathtt{R}\text{-}\mathtt{MP}$): & $\dfrac{\phi, \phi \rightarrow \psi}
{\psi}$ \vspace{5pt} \\
Conditional necessitation ($\mathtt{R}\text{-}\mathtt{CN}$): & $\dfrac{\phi}
{\Fclo{\FAA} \psi \rightarrow \Fclo{\emptyset} \phi}$
\end{tabular}

\end{definition}

The differences between the axiomatic system given in Definition \ref{definition:An axiomatic system for CL} and the one given in Definition \ref{definition:Another axiomatic system for CL} are that the former does not have the axioms $\mathtt{A}\text{-}\mathtt{MG}$, $\mathtt{A}\text{-}\mathtt{MC}$, $\mathtt{A}\text{-}\mathtt{Det}$, and the rule $\mathtt{R}\text{-}\mathtt{CN}$, and the latter does not have the axioms $\mathtt{A}\text{-}\mathtt{Mon}$, $\mathtt{A}\text{-}\mathtt{Max}$, and the rule $\mathtt{R}\text{-}\mathtt{RE}$.

\begin{fact}
The axiomatic system given in Definition \ref{definition:An axiomatic system for CL} is equivalent to the axiomatic system given in Definition \ref{definition:Another axiomatic system for CL}.
\end{fact}

\begin{proof}
~

We use $\vdash_1 \phi$ to indicate that $\phi$ is derivable in the former system and $\vdash_2 \phi$ to indicate that $\phi$ is derivable in the latter system.

First, we show that all the axioms and rules of the latter system are derivable in the former system. It suffices to show that the axioms $\mathtt{A}\text{-}\mathtt{MG}$, $\mathtt{A}\text{-}\mathtt{MC}$, $\mathtt{A}\text{-}\mathtt{Det}$, and the rule $\mathtt{R}\text{-}\mathtt{CN}$ are derivable in the former system.
\begin{itemize}

\item

$\mathtt{A}\text{-}\mathtt{MG}$.
It suffices to show $\vdash_1 (\Fclo{\emptyset} (\phi \rightarrow \psi) \land \Fclo{\FAA} \phi) \rightarrow \Fclo{\FAA} \psi$.
By Axiom $\mathtt{A}\text{-}\mathtt{IA}$, $\vdash_1 (\Fclo{\emptyset} (\phi \rightarrow \psi) \land \Fclo{\FAA} \phi) \rightarrow \Fclo{\FAA} ((\phi \rightarrow \psi) \land \phi)$.
By Rule $\mathtt{R}\text{-}\mathtt{RE}$, $\vdash_1 \Fclo{\FAA} ((\phi \rightarrow \psi) \land \phi) \rightarrow \Fclo{\FAA} (\phi \land \psi)$.
By Axiom $\mathtt{A}\text{-}\mathtt{Mon}$, $\vdash_1 \Fclo{\FAA} (\phi \land \psi) \rightarrow \Fclo{\FAA} \psi$.
Then, $\vdash_1 (\Fclo{\emptyset} (\phi \rightarrow \psi) \land \Fclo{\FAA} \phi) \rightarrow \Fclo{\FAA} \psi$.

\item

$\mathtt{A}\text{-}\mathtt{MC}$.
Assume $\FAA \subseteq \FBB$. We want to show $\vdash_1 \Fclo{\FAA} \phi \rightarrow \Fclo{\FBB} \phi$.
By Axiom $\mathtt{A}\text{-}\mathtt{Ser}$, $\vdash_1 \Fclo{\FBB-\FAA} \top$. Then, $\vdash_1 \Fclo{\FAA} \phi \rightarrow (\Fclo{\FAA} \phi \land \Fclo{\FBB-\FAA} \top)$.
By Axiom $\mathtt{A}\text{-}\mathtt{IA}$, $\vdash_1 \Fclo{\FAA} \phi \land \Fclo{\FBB-\FAA} \top \rightarrow \Fclo{\FBB} (\phi \land \top)$.
By Rule $\mathtt{R}\text{-}\mathtt{RE}$, $\vdash_1 \Fclo{\FBB} (\phi \land \top) \rightarrow \Fclo{\FBB} \phi$.
Then we can get $\vdash_1 \Fclo{\FAA} \phi \rightarrow \Fclo{\FBB} \phi$.

\item 

$\mathtt{A}\text{-}\mathtt{Det}$. We want to show $\vdash_1 \Fclo{\FAA} (\phi \lor \psi) \rightarrow (\Fclo{\FAA} \phi \lor \Fclo{\FAG} \psi)$.
By Axiom $\mathtt{A}\text{-}\mathtt{IA}$, $\vdash_1 (\Fclo{\FAA} (\phi \lor \psi) \land \Fclo{\emptyset} \neg \psi) \rightarrow \Fclo{\FAA} ((\phi \lor \psi) \land \neg \psi)$.
By Rule $\mathtt{R}\text{-}\mathtt{RE}$, $\vdash_1 \Fclo{\FAA} ((\phi \lor \psi) \land \neg \psi) \rightarrow \Fclo{\FAA} (\phi \land \neg \psi)$.
By Axiom $\mathtt{A}\text{-}\mathtt{Mon}$, $\vdash_1 \Fclo{\FAA} (\phi \land \neg \psi) \rightarrow \Fclo{\FAA} \phi$.
Then, $\vdash_1 (\Fclo{\FAA} (\phi \lor \psi) \land \Fclo{\emptyset} \neg \psi) \rightarrow \Fclo{\FAA} \phi$.
By Axiom $\mathtt{A}\text{-}\mathtt{Max}$, $\vdash_1 \neg \Fclo{\FAG} \psi \rightarrow \Fclo{\emptyset} \neg \psi$.
Then, $\vdash_1 (\Fclo{\FAA} (\phi \lor \psi) \land \neg \Fclo{\FAG} \psi) \rightarrow \Fclo{\FAA} \phi$.
Then, $\vdash_1 (\Fclo{\FAA} (\phi \lor \psi) \rightarrow (\neg \Fclo{\FAG} \psi \rightarrow \Fclo{\FAA} \phi)$.
Then, $\vdash_1 \Fclo{\FAA} (\phi \lor \psi) \rightarrow (\Fclo{\FAA} \phi \lor \Fclo{\FAG} \psi)$.

\item

$\mathtt{R}\text{-}\mathtt{CN}$.
Assume $\vdash_1 \phi$. We want to show $\vdash_1 \Fclo{\FAA} \psi \rightarrow \Fclo{\emptyset} \phi$. Note $\vdash_1 \top \leftrightarrow \phi$. By Rule $\mathtt{R}\text{-}\mathtt{RE}$, $\vdash_1 \Fclo{\emptyset} \top \rightarrow \Fclo{\emptyset} \phi$.
By Axiom $\mathtt{A}\text{-}\mathtt{Ser}$, $\vdash_1 \Fclo{\emptyset} \top$. Then $\vdash_1 \Fclo{\emptyset} \phi$. Then $\vdash_1 \Fclo{\FAA} \psi \rightarrow \Fclo{\emptyset} \phi$.

\end{itemize}

Second, we show that all the axioms and rules of the former system are derivable in the latter system. It suffices to show that the axioms $\mathtt{A}\text{-}\mathtt{Mon}$, $\mathtt{A}\text{-}\mathtt{Max}$, and the rule $\mathtt{R}\text{-}\mathtt{RE}$ are drivable in the latter system.

\begin{itemize}

\item

$\mathtt{A}\text{-}\mathtt{Mon}$.
Note $\vdash_2 (\phi \land \psi) \rightarrow \phi$.
By Rule $\mathtt{R}\text{-}\mathtt{CN}$, $\vdash_2 \Fclo{\FAA} (\phi \land \psi) \rightarrow \Fclo{\emptyset} ((\phi \land \psi) \rightarrow \phi)$.
By Axiom $\mathtt{A}\text{-}\mathtt{MG}$, $\vdash_2 \Fclo{\emptyset} ((\phi \land \psi) \rightarrow \phi) \rightarrow (\Fclo{\FAA} (\phi \land \psi) \rightarrow \Fclo{\FAA} \phi)$.
Then $\vdash_2 \Fclo{\FAA} (\phi \land \psi) \rightarrow (\Fclo{\FAA} (\phi \land \psi) \rightarrow \Fclo{\FAA} \phi)$.
Then $\vdash_2 \Fclo{\FAA} (\phi \land \psi) \rightarrow \Fclo{\FAA} \phi$.

\item 

$\mathtt{A}\text{-}\mathtt{Max}$.
Note $\vdash_2 \neg \phi \lor \phi$. By Rule $\mathtt{R}\text{-}\mathtt{CN}$, $\vdash_2 \Fclo{\emptyset} \top \rightarrow \Fclo{\emptyset} (\neg \phi \lor \phi)$. 
By Axiom $\mathtt{A}\text{-}\mathtt{Ser}$, $\vdash_2 \Fclo{\emptyset} \top$.
Then $\vdash_2 \Fclo{\emptyset} (\neg \phi \lor \phi)$.
Note $\vdash_2 \Fclo{\emptyset} (\neg \phi \lor \phi) \rightarrow (\Fclo{\emptyset} \neg \phi \lor \Fclo{\FAG} \phi)$ is an instance of Axiom $\mathtt{A}\text{-}\mathtt{Det}$. Then $\vdash_2 \Fclo{\emptyset} \neg \phi \lor \Fclo{\FAG} \phi$. Then, $\vdash_2 \neg \Fclo{\emptyset} \neg \phi \rightarrow \Fclo{\FAG} \phi$.

\item

$\mathtt{R}\text{-}\mathtt{RE}$.
Assume $\vdash_2 \phi \leftrightarrow \psi$.
We want to show $\vdash_2 \Fclo{\FAA} \phi \leftrightarrow \Fclo{\FAA} \psi$.
By Rule $\mathtt{R}\text{-}\mathtt{CN}$, $\vdash_2 \Fclo{\FAA} \phi \rightarrow \Fclo{\emptyset} (\phi \rightarrow \psi)$.
By Axiom $\mathtt{A}\text{-}\mathtt{MG}$, $\vdash_2 \Fclo{\emptyset} (\phi \rightarrow \psi) \rightarrow (\Fclo{\FAA} \phi \rightarrow \Fclo{\FAA} \psi)$.
Then, $\vdash_2 \Fclo{\FAA} \phi \rightarrow (\Fclo{\FAA} \phi \rightarrow \Fclo{\FAA} \psi)$.
Then $\vdash_2 \Fclo{\FAA} \phi \rightarrow \Fclo{\FAA} \psi$.
Similarly, we can show $\vdash_2 \Fclo{\FAA} \psi \rightarrow \Fclo{\FAA} \phi$.
Then $\vdash_2 \Fclo{\FAA} \phi \leftrightarrow \Fclo{\FAA} \psi$.

\end{itemize}

\end{proof}

\section{On seriality, independence of agents, and determinism}
\label{section:Three assumptions in concurrent game models are too strong}

As mentioned, concurrent game models have the assumptions of \emph{seriality}, \emph{the independence of agents}, and \emph{determinism}. In this section, against the literature, we argue that the three assumptions do not hold generally. In addition, we also discuss the potential impact of dropping them.

\subsection{Coalitions do not have available joint actions in all situations}

The real world evolves continuously, but many important \emph{artificial settings}, such as games, have boundaries, including terminal states. For example, a paper-scissors-rock game terminates when a winner is decided. Intuitively, in terminating states of these artificial settings, agents do not have available actions.

One natural thought is to model termination by introducing a special action, that is, ``\emph{doing nothing}'', which does not change anything. Terminal states are those where every agent can only perform the special action. We do not see that this causes a problem, although it seems conceptually incorrect. Note that some literature, such as \cite{AlechinaLoganNgaRakib2011}, acknowledges that ``doing nothing'' may lead to state change.

Dropping seriality might change things. Take Alternating-time Temporal Logic $\FATL$ as an example, whose models are concurrent game models. In $\FATL$, \emph{paths}, that are infinite sequences of states, play a crucial role. If we drop seriality, we shall consider both finite paths and infinite paths.

\subsection{Available joint actions are not unconditional in all situations}

Agents' actions often interfere with each other, and whether a coalition can perform an action is often conditional on other agents' actions at the same time.
The interference among agents' actions can be due to many reasons, including a lack of resources. It fails the independence of agents. If we want to consider interactions among agents' actions, we should relax the assumption of agent independence.

What follows are three examples where the independence of agents does not hold; the last two are from Sergot \cite{sergot_examples_2014}.

\begin{example}
\label{example:two persons one chair}

In a room, there are two agents, $a$ and $b$, and only one chair. Agent $a$ can sit, agent $b$ can sit, but they cannot sit at the same time.

\end{example}

\begin{example}
\label{example:two persons two rooms pass}

There are two rooms, separated by a doorway, and two agents, $a$ and $b$, in the room on the left. The two agents can stay where they are or pass from one room to another, but not simultaneously, as the doorway is too narrow.

\end{example}

\begin{example}
\label{example:two persons a vase move}

There are two agents, $a$ and $b$, and a precious vase. The vase can be in one of
three locations: $\mathtt{in}$, $\mathtt{out}$, and $\mathtt{over}$. It might be raining or not, which is out of the control of $a$ and $b$. It is wrong for the vase to be $\mathtt{out}$ when it is raining.
Agent $a$ can move the vase between $\mathtt{in}$ and $\mathtt{out}$, and agent $b$ can move it between $\mathtt{out}$ and $\mathtt{over}$.
The vase is in $\mathtt{out}$, and it is not raining. In this case, agent $a$ can move the vase to $\mathtt{in}$, agent $b$ can move the vase to $\mathtt{over}$, but they cannot do it simultaneously.

\end{example}

One natural thought is to treat the actions in the scenarios given above as \emph{trying to do something}. Since people can always \emph{try} to do things at the same time, the independence of agents holds. There are some issues with this idea.
First, it would trivialize the notion of availability: all joint actions would be available. However, availability is not a trivial notion. For example, it makes perfect sense to say ``\emph{they cannot sit at the same time}'' in the first example.
Second, this thought seems incompatible with determinism. For example, it is natural to think that the outcome of $a$ and $b$ trying to sit is not deterministic in the first example.
Finally, Royakkers and Hughes \cite{royakkers_blame_2020} found that this idea can lead to implausible conclusions in certain situations regarding whether an agent is \emph{accountable} for a result.

There are advantages to dropping the assumption of agent independence: it lets us model \emph{conditional abilities}. Roughly, conditional abilities are something like this: \emph{given that the coalition $\FAA$ is going to behave in a certain way, the coalition $\FBB$ is able to ensure a certain goal}. We refer to \cite{goranko_logic_2022} for more discussion of conditional abilities.

\subsection{Joint actions of the grand coalition do not have a unique outcome in all situations}

There are many situations where some joint actions of all behaving agents have more than one outcome state. What follows are some examples; the last two of them are from \cite{sergot_examples_2014}.

\begin{example}
\label{example:two persons two dices}

Two agents, $a$ and $b$, are playing a simple game. They throw a dice simultaneously, and whoever gets a bigger number wins.

\end{example}

\begin{example}
\label{example:two persons a desk a vase}

A vase stands on a table. There is an agent $a$ who can lift or lower the end of the table. If the table tilts, the vase might fall, and if it falls, it might break.

\end{example}

\begin{example}
\label{example:two persons two rooms try to pass}

As in Example \ref{example:two persons two rooms pass}, there are two rooms, separated by a narrow doorway, and two agents, $a$ and $b$, in the room on the left.
The agents can stay where they are or \emph{try to pass} from one room to the other.
If both try, either both fail and stay in the room on the left, or $a$ succeeds in moving through because $a$ is a little stronger than $b$.

\end{example}

Furthermore, in many situations, we consider some but not all behaving agents, who form the grand coalition. The reason could be that we do not know exactly who is behaving, or that we know who is behaving, but only some of them are important to us. In these situations, determinism fails.

A natural idea is to suppose an implicit agent, \emph{the environment}, which represents \emph{nature}, unknown agents, and ignored agents as a whole.
About this idea, let us quote some words from Sergot \cite{sergot_actual_2022}: ``\emph{In many examples however the device of thinking of the environment as a kind of agent is very much less convincing. Suppose that $a$ and $b$ can throw the vase out of a window. If it is thrown out of the window and breaks, the effects of the environment are to contribute—what? fragility? the angle at which the vase strikes the ground? In examples, common in the literature on causation, where agents throw rocks at a bottle with imperfect aim, what is the environment's contribution? To affect the trajectory of the throw? To affect the trajectory of each throw differently? To affect the trajectory of each throw differently depending on who else throws at the same time? To contribute fragility to the bottle? In all but the very simplest examples the device seems extremely contrived, and does not scale.}''

\paragraph{Remarks}

The granularity of actions can play a role in whether determinism holds.
For example, in Example \ref{example:two persons two rooms try to pass}, we could specify which ways of trying to pass can let $a$ pass through and which ways cannot. Then, the joint actions of $a$ and $b$ would be deterministic.

\section{A minimal coalition logic $\FMCL$ based on general concurrent game models}
\label{section:A minimal coalition logic MCL based on general concurrent game models}

In this section, we define models of $\FMCL$, axiomatize $\FMCL$, argue that $\FMCL$ is a minimal coalition logic for reasoning about coalitional powers, and discuss whether $\FCL$ can simulate $\FMCL$.

\subsection{General concurrent game models}

\begin{definition}[General concurrent game models]
\label{definition:General concurrent game models}

An abstract multi-agent action model $\MM = (\FST, \FAC, \{\Fav_\FAA \mid \FAA \subseteq \FAG\}, \{\Fout_\FAA \mid \FAA \subseteq \FAG\}, \Flab)$ is a \Fdefs{general concurrent game model}, where:
\begin{enumerate}[label=(\arabic*),leftmargin=3.33em]

\item

for every $\FAA \subseteq \FAG$ and $s \in \FST$, $\Fav_\FAA (s) = \Fav_\FAG (s)|_\FAA$.

\emph{Intuitively, $\Fav_\FAA$ is determined by $\Fav_\FAG$}.

\item

for every $s \in \FST$, $\Fav_\FAG (s) = \{\sigma_\FAG \in \FJA_\FAG \mid \Fout_\FAG (s, \sigma_\FAG) \neq \emptyset \}$.

\item

for every $\FAA \subseteq \FAG$, $s \in \FST$ and $\sigma_\FAA \in \FJA_\FAA$, $\Fout_\FAA (s, \sigma_\FAA) = \bigcup \{\Fout_\FAG (s, \sigma_\FAG) \mid \sigma_\FAG \in \FJA_\FAG \text{ and}$ $\sigma_\FAA \subseteq \sigma_\FAG\}$.

\end{enumerate}

\end{definition}


The following fact offers an equivalent definition of general concurrent game models.

\begin{fact}
\label{definition:alternative definition of general concurrent game models}

Let $\MM = (\FST, \FAC, \{\Fav_\FAA \mid \FAA \subseteq \FAG\}, \{\Fout_\FAA \mid \FAA \subseteq \FAG\}, \Flab)$ be an abstract multi-agent action model.
The following two sets of conditions are equivalent:
\begin{enumerate}[label=(\arabic*),leftmargin=3.33em]

\item 
\begin{enumerate}

\item

for every $\FAA \subseteq \FAG$ and $s \in \FST$, $\Fav_\FAA (s) = \Fav_\FAG (s)|_\FAA$;

\item

for every $s \in \FST$, $\Fav_\FAG (s) = \{\sigma_\FAG \in \FJA_\FAG \mid \Fout_\FAG (s, \sigma_\FAG) \neq \emptyset \}$;

\item

for every $\FAA \subseteq \FAG$, $s \in \FST$ and $\sigma_\FAA \in \FJA_\FAA$, $\Fout_\FAA (s, \sigma_\FAA) = \bigcup \{\Fout_\FAG (s, \sigma_\FAG) \mid \sigma_\FAG \in \FJA_\FAG \text{ and}$ $\sigma_\FAA \subseteq \sigma_\FAG\}$.

\end{enumerate}

\item

\begin{enumerate}

\item

for every $\FAA \subseteq \FAG$, $s \in \FST$ and $\sigma_\FAA \in \FJA_\FAA$, $\Fout_\FAA (s, \sigma_\FAA) = \bigcup \{\Fout_\FAG (s, \sigma_\FAG) \mid \sigma_\FAG \in \FJA_\FAG \text{ and}$ $\sigma_\FAA \subseteq \sigma_\FAG\}$;

\item

for every $\FAA \subseteq \FAG$ and $s \in \FST$, $\Fav_\FAA (s) = \{\sigma_\FAA \in \FJA_\FAA \mid \Fout_\FAA (s, \sigma_\FAA) \neq \emptyset\}$.

\end{enumerate}

\end{enumerate}

\end{fact}

\begin{proof}
~

(1) $\Rightarrow$ (2)

Assume (1).
We only need to show (2b). Let $\FAA \subseteq \FAG$, $s \in \FST$.

Assume $\sigma_\FAA \in \Fav_\FAA (s)$. By (2a), there is $\sigma_\FAG \in \Fav_\FAG (s)$ such that $\sigma_\FAA \subseteq \sigma_\FAG$. By (1b), $\Fout_\FAG (s,\sigma_\FAG) \neq \emptyset$. By (1c), $\Fout_\FAA (s, \sigma_\FAA) \neq \emptyset$.

Assume $\Fout_\FAA (s, \sigma_\FAA) \neq \emptyset$. By (1c), there is $\sigma_\FAG \in \FJA_\FAG$ such that $\sigma_\FAA \subseteq \sigma_\FAG$ and $\Fout_\FAG (s, \sigma_\FAG) \neq \emptyset$. By (1b), $\sigma_\FAG \in \Fav_\FAG (s)$. By (1a), $\sigma_\FAA \in \Fav_\FAA (s)$.

\medskip

(2) $\Rightarrow$ (1)

Assume (2). We only need to show (1a).
Let $\FAA \subseteq \FAG$ and $s \in \FST$.

Assume $\sigma_\FAA \in \Fav_\FAA (s)$. By (2b), $\Fout_\FAA (s,\sigma_\FAA) \neq \emptyset$. By (2a), $\Fout_\FAG (s,\sigma_\FAG) \neq \emptyset$ for some $\sigma_\FAG \in \FJA_\FAG$ such that $\sigma_\FAA \subseteq \sigma_\FAG$. By (2b), $\sigma_\FAG \in \Fav_\FAG (s)$. Then $\sigma_\FAA \in \Fav_\FAG (s)|_\FAA$.

Assume $\sigma_\FAA \in \Fav_\FAG (s)|_\FAA$. Then $\sigma_\FAA \subseteq \sigma_\FAG$ for some $\sigma_\FAG \in \Fav_\FAG (s)$. By (2b), $\Fout_\FAG (s,\sigma_\FAG) \neq \emptyset$. By (2a), $\Fout_\FAA (s,\sigma_\FAA) \neq \emptyset$. By (2b), $\sigma_\FAA \in \Fav_\FAA (s)$.

\end{proof}

By this fact, availability functions in general concurrent game models are also coincident with our understanding of available joint actions given before.

The following fact gives two equivalent sets of constraints on availability functions in general concurrent game models, from where we can see that available joint actions in general concurrent game models are \emph{conditional}:

\begin{fact}
\label{fact:Characterization of available joint action functions of general concurrent game models}

Let $\MM = (\FST, \FAC, \{\Fav_\FAA \mid \FAA \subseteq \FAG\}, \{\Fout_\FAA \mid \FAA \subseteq \FAG\}, \Flab)$ be an abstract multi-agent action model.
Then, the following three conditions are equivalent:
\begin{enumerate}[label=(\arabic*),leftmargin=3.33em]

\item

for every $\FAA \subseteq \FAG$ and $s \in \FST$, $\Fav_\FAA (s) = \Fav_\FAG (s)|_\FAA$.

\item

for every $\FAA, \FBB \subseteq \FAG$ such that $\FAA \cap \FBB = \emptyset$, and $s \in \FST$:
\begin{enumerate}

\item 

for every $\sigma_{\FAA \cup \FBB} \in \Fav_{\FAA \cup \FBB} (s)$, $\sigma_{\FAA \cup \FBB}|_\FAA \in \Fav_\FAA (s)$;

\item 

for every $\sigma_\FAA \in \Fav_\FAA (s)$, there is $\sigma_\FBB \in \Fav_\FBB (s)$ such that $\sigma_\FAA \cup \sigma_\FBB \in \Fav_{\FAA \cup \FBB} (s)$.

\end{enumerate}

\item

for every $\FAA \subseteq \FAG$ and $s \in \FST$:
\begin{enumerate}

\item 

for every $\sigma_\FAG \in \Fav_\FAG (s)$, $\sigma_\FAG|_\FAA \in \Fav_\FAA (s)$;

\item 

for every $\sigma_\FAA \in \Fav_\FAA (s)$, there is $\sigma_\FAAb \in \Fav_{\FAAb} (s)$ such that $\sigma_\FAA \cup \sigma_\FAAb \in \Fav_\FAG (s)$.

\end{enumerate}

\end{enumerate}

\end{fact}

\begin{proof}
~

(1) $\Rightarrow$ (2)

Assume (1). Let $\FAA, \FBB \subseteq \FAG$ be such that $\FAA \cap \FBB = \emptyset$, and $s \in \FST$.

First, we show (2a).
Let $\sigma_{\FAA \cup \FBB} \in \Fav_{\FAA \cup \FBB} (s)$.
By (1), there is $\FAja{\FAG} \in \FAajas{\FAG}$ such that $\sigma_{\FAA \cup \FBB} \subseteq \FAja{\FAG}$. By (1), $\FAja{\FAG}|_\FAA \in \FAajas{\FAA}$. Note $\FAja{\FAG}|_\FAA = \FAja{\FAA \cup \FBB}|_\FAA$. Then $\FAja{\FAA \cup \FBB}|_\FAA \in \FAajas{\FAA}$.

Second, we show (2b).
Let $\FAja{\FAA} \in \FAajas{\FAA}$. We want to show that there is $\sigma_\FBB \in \Fav_\FBB (s)$ such that $\sigma_\FAA \cup \sigma_\FBB \in \Fav_{\FAA \cup \FBB} (s)$.
By (1), there is $\FAja{\FAG} \in \FAajas{\FAG}$ such that $\sigma_\FAA \subseteq \FAja{\FAG}$. By (1), $\FAja{\FAG}|_\FBB \in \FAajas{\FBB}$ and $\FAja{\FAG}|_{\FAA \cup \FBB} \in \FAajas{\FAA \cup \FBB}$. Note $\FAja{\FAG}|_{\FAA \cup \FBB} = \FAja{\FAA} \cup \FAja{\FAG}|_\FBB$. Then, $\FAja{\FAA} \cup \FAja{\FAG}|_\FBB \in \FAajas{\FAA \cup \FBB}$.

(2) $\Rightarrow$ (3)

It is easy to see that (3) is a special case of (2).

(3) $\Rightarrow$ (1)

Assume (3). Let $\FAA \subseteq \FAG$ and $s \in \FST$.

Let $\FAja{\FAA} \in \FAajas{\FAA}$. By (3b), there is $\FAja{\FAG} \in \FAajas{\FAG}$ such that $\FAja{\FAA} \subseteq \FAja{\FAG}$. Then $\FAja{\FAA} \in \FAajas{\FAG}|_\FAA$.

Let $\FAja{\FAA} \in \FAajas{\FAG}|_\FAA$. Then there is $\FAja{\FAG} \in \FAajas{\FAG}$ such that $\FAja{\FAA} \subseteq \FAja{\FAG}$. By (3a), $\FAja{\FAG}|_\FAA \in \FAajas{\FAA}$, that is, $\FAja{\FAA} \in \FAajas{\FAA}$.

\end{proof}

From the facts \ref{fact:Characterization of available joint action functions of concurrent game models} and \ref{fact:Characterization of available joint action functions of general concurrent game models}, we can see that in concurrent game models, for all $\sigma_\FAA \in \Fav_\FAA (s)$, for all $\sigma_\FAAb \in \Fav_{\FAAb} (s)$, $\sigma_\FAA \cup \sigma_\FAAb \in \Fav_\FAG (s)$,
while in general concurrent game models, for all $\sigma_\FAA \in \Fav_\FAA (s)$, there is $\sigma_\FAAb \in \Fav_{\FAAb} (s)$ such that $\sigma_\FAA \cup \sigma_\FAAb \in \Fav_\FAG (s)$.


What follows is an example of how scenarios are represented by pointed general concurrent game models, which differs from Example \ref{example:one mask} in some important respects.

\begin{example}
\label{example:one mask}
Adam and Bob work in a closed gas laboratory when an earthquake occurs. They find that the door cannot be opened.

One toxic gas container is damaged and \emph{might} leak soon. Unfortunately, only one gas mask is available in the laboratory.

Adam and Bob have two available actions in this situation: \emph{doing nothing} and \emph{wearing a gas mask}. If they both do nothing, they will both die or survive. If Adam does nothing and Bob wears the gas mask, Adam will die or not, and Bob will survive. If Adam wears the gas mask and Bob does nothing, Adam will survive, and Bob will die or not.

We suppose that one can only do nothing once he wears a gas mask; Adam and Bob have no available joint actions once one of them dies.
This situation can be represented by the pointed general concurrent game model depicted in Figure \ref{figure:one mask}.
\end{example}

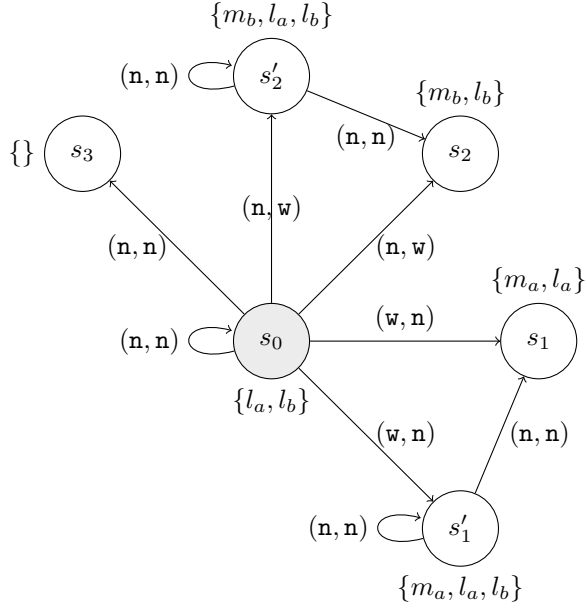
\begin{figure}
\begin{center}
\begin{tikzpicture}
[
->=stealth,
scale=1,
every node/.style={transform shape},
]

\tikzstyle{every state}=[minimum size=10mm]

\node[state,fill=gray!15] (s-0) {$s_0$};
\node[state,position=0:{25mm} from s-0] (s-1) {$s_1$};
\node[state,position=45:{25mm} from s-0] (s-2) {$s_2$};
\node[state,position=135:{25mm} from s-0] (s-3) {$s_3$};
\node[state,position=315:{25mm} from s-0] (sp-1) {$s'_1$};
\node[state,position=90:{25mm} from s-0] (sp-2) {$s'_2$};

\node[below=5mm] (a-s-0) at (s-0) {$\{l_a,l_b\}$};
\node[above=5mm] (a-s-1) at (s-1) {$\{m_a,l_a\}$};
\node[above=5mm] (a-s-2) at (s-2) {$\{m_b,l_b\}$};
\node[left=5mm] (a-s-3) at (s-3) {$\{\}$};
\node[below=5mm] (a-sp-1) at (sp-1) {$\{m_a,l_a,l_b\}$};
\node[above=5mm] (a-sp-2) at (sp-2) {$\{m_b,l_a,l_b\}$};

\path
(s-0) edge [above] node {$\mathtt{(w,n)}$} (s-1)
(s-0) edge [right] node {$\mathtt{(w,n)}$} (sp-1)
(s-0) edge [right] node {$\mathtt{(n,w)}$} (s-2)
(s-0) edge [midway] node {$\mathtt{(n,w)}$} (sp-2)
(s-0) edge [left] node {$\mathtt{(n,n)}$} (s-3)
(s-0) edge [loop left] node {$\mathtt{(n,n)}$} (s-0)
(sp-1) edge [right] node {$\mathtt{(n,n)}$} (s-1)
(sp-1) edge [loop left] node {$\mathtt{(n,n)}$} (sp-1)
(sp-2) edge [below] node {$\mathtt{(n,n)}$} (s-2)
(sp-2) edge [loop left] node {$\mathtt{(n,n)}$} (sp-2)
;

\end{tikzpicture}

\caption{
This figure indicates a pointed general concurrent game model $(\MM,s_0)$ for the situation of Example \ref{example:one mask}.
Again, here: $m_a, m_b, l_a$ and $l_b$ respectively express the propositions \emph{Adam is wearing a gas mask}, \emph{Bob is wearing a gas mask}, \emph{Adam is alive} and \emph{Bob is alive}; $\mathtt{n}$ and $\mathtt{w}$ respectively represent the actions \emph{doing nothing} and \emph{wearing a mask}; an arrow from a state $x$ to a state $y$ labeled with an action profile $\alpha$, indicates that $y$ is a possible outcome state of performing $\alpha$ at $x$.
Note that at states $s_1, s_2$ and $s_3$, neither Bob nor Adam has an available action. So, this general concurrent game model is not serial.
Note that this general concurrent game model does not meet the independence of agents: the action $\mathtt{w}$ is available for Adam and also available for Bob at $s_0$, but the action $(\mathtt{w},\mathtt{w})$ is not available for Adam and Bob as a coalition.
Note that this general concurrent game model is not deterministic: none of the action profiles $(\mathtt{w},\mathtt{n}), (\mathtt{n},\mathtt{w})$ and $(\mathtt{n},\mathtt{n})$ has a unique outcome state at $s_0$.
}

\label{figure:one mask}

\end{center}
\end{figure}

\subsection{Language and semantics}

In $\FMCL$, the formula $\Fclo{\FAA} \phi$ intuitively indicates that \emph{some conditionally available joint action of $\FAA$ ensures $\phi$}, and $\Fclod{\FAA} \phi$ intuitively indicates \emph{every conditionally available joint action of $\FAA$ enables $\phi$}.

Note that $\Box \phi$ is defined as $\Fclo{\emptyset} \top \rightarrow \Fclo{\emptyset} \phi$ and $\Diamond \phi$ is defined as $\Fclo{\emptyset} \top \land \Fclod{\emptyset} \phi$.
In general concurrent game models, the empty action is available for the empty coalition if and only if it has an outcome state. Consequently:

\medskip

\begin{tabular}{lll}
$\MM, s \Vdash \Box \phi$ & $\Leftrightarrow$ & \parbox[t]{28em}{for all $t \in \Fout_\emptyset (s, \emptyset)$, $\MM, t \Vdash \phi$} \\
$\MM, s \Vdash \Diamond \phi$ & $\Leftrightarrow$ & \parbox[t]{28em}{there is $t \in \Fout_\emptyset (s, \emptyset)$ such that $\MM, t \Vdash \phi$}
\end{tabular}

\medskip

We use $\models_\FMCL \phi$ to indicate that $\phi$ is $\FMCL$-\emph{valid}: $\phi$ is true at all pointed general concurrent game models. We often drop ``$\FMCL$'' when the contexts are clear.

\medskip

We briefly discuss some differences between $\FCL$-validity and $\FMCL$-validity. 

\begin{fact}
$(\Fclo{\FAA} \phi \land \Fclo{\FBB} \psi) \rightarrow \Fclo{\FAA \cup \FBB} (\phi \land \psi)$, where $\FAA \cap \FBB = \emptyset$, is $\FCL$-valid but not $\FMCL$-valid.
\end{fact}


This formula is intuitively related to the independence of agents. Since models of $\FMCL$ might not satisfy this assumption, this fact is expected.

As mentioned above, available actions in $\FCL$ mean that they are unconditionally performable.
Given this understanding, the $\FCL$-validity of $(\Fclo{\FAA} \phi \land \Fclo{\FBB} \psi) \rightarrow \Fclo{\FAA \cup \FBB} (\phi \land \psi)$ becomes clear:
suppose that the coalition $\FAA$ has an absolutely available joint action $\sigma_\FAA$ to ensure $\phi$, and the coalition $\FBB$ has an absolutely available joint action $\sigma_\FBB$ to ensure $\psi$;
then, $\sigma_\FAA \cup \sigma_\FBB$ is absolutely available for $\FAA \cup \FBB$, which ensures $\phi$ and $\psi$.

In $\FMCL$, available actions mean that they are conditionally performable. Given this understanding, the $\FMCL$-invalidity of $(\Fclo{\FAA} \phi \land \Fclo{\FBB} \psi) \rightarrow \Fclo{\FAA \cup \FBB} (\phi \land \psi)$ is also understandable:
suppose that the coalition $\FAA$ has only one conditionally available joint action $\sigma_\FAA$ to ensure $\phi$, the coalition $\FBB$ has only one conditionally available joint action $\sigma_\FBB$ to ensure $\psi$, but the availability of $\sigma_\FAA$ and $\sigma_\FBB$ is not conditional on each other and $\sigma_\FAA \cup \sigma_\FBB$ is not available for $\FAA \cup \FBB$; thus, $\FAA \cup \FBB$ is unable to ensure $\phi$ and $\psi$.

\begin{fact}
$\neg \Fclo{\emptyset} \neg \phi \rightarrow \Fclo{\FAG} \phi$ and $\Fclo{\FAG} \phi \lor \Fclo{\FAG} \neg \phi$ are $\FCL$-valid but not $\FMCL$-valid.
\end{fact}

The two formulas are intuitively related to determinism. Since models of $\FMCL$ might not be deterministic, this fact is expected.

Note that relaxing seriality also causes the two formulas to be invalid.
Consider a pointed general concurrent game model $(\MM,s)$ where every coalition has no available joint action. Then, $\neg \Fclo{\emptyset} \neg \phi$ is trivially true, and $\Fclo{\FAG} \phi$ and $\Fclo{\FAG} \neg \phi$ are trivially false at $(\MM,s)$. Then, $\neg \Fclo{\emptyset} \neg \phi \rightarrow \Fclo{\FAG} \phi$ and $\Fclo{\FAG} \phi \lor \Fclo{\FAG} \neg \phi$ are trivially false at $(\MM,s)$.

\subsection{Axiomatization}

\begin{definition}[An axiomatic system for $\FMCL$]
\label{definition:An axiomatic system for MCL}
~

\noindent \emph{Axioms}:

\medskip

\begin{tabular}{rl}
Tautologies ($\mathtt{A}\text{-}\mathtt{Tau}$): & all propositional tautologies \vspace{5pt} \\
Monotonicity of goals ($\mathtt{A}\text{-}\mathtt{MG}$): & $\Fclo{\emptyset} (\phi \rightarrow \psi) \rightarrow (\Fclo{\FAA} \phi \rightarrow \Fclo{\FAA} \psi)$ \vspace{5pt} \\
Monotonicity of coalitions ($\mathtt{A}\text{-}\mathtt{MC}$): & $\Fclo{\FAA} \phi \rightarrow \Fclo{\FBB} \phi$, where $\FAA \subseteq \FBB$ \vspace{5pt} \\
Liveness ($\mathtt{A}\text{-}\mathtt{Live}$): & $\neg \Fclo{\FAA} \bot$
\end{tabular}

\medskip

\noindent \emph{Inference rules}:

\medskip

\begin{tabular}{rl}
Modus ponens ($\mathtt{R}\text{-}\mathtt{MP}$): & $\dfrac{\phi, \phi \rightarrow \psi}
{\psi}$ \vspace{5pt} \\
Conditional necessitation ($\mathtt{R}\text{-}\mathtt{CN}$): & $\dfrac{\phi}
{\Fclo{\FAA} \psi \rightarrow \Fclo{\emptyset} \phi}$
\end{tabular}

\end{definition}

We use $\vdash_\FMCL \phi$ to indicate $\phi$ is \Fdefs{derivable} in this system. We often drop ``$\FMCL$'' when contexts are clear.

The soundness and completeness of the system are shown in Section \ref{section:Completeness of MCL}.

The only difference between the axiomatic system for $\FCL$ given in Definition \ref{definition:Another axiomatic system for CL} and the axiomatic system for $\FMCL$ is that the following three axioms are missing in the latter, which are intuitively related to the three assumptions of concurrent game models missing in general concurrent game models\footnote{Is there a precise sense in which the three axioms are related to the three assumptions? We leave investigations of this issue for further work.}:

\medskip

\begin{tabular}{rl}
Independence of agents ($\mathtt{A}\text{-}\mathtt{IA}$): & $(\Fclo{\FAA} \phi \land \Fclo{\FBB} \psi) \rightarrow \Fclo{\FAA \cup \FBB} (\phi \land \psi)$, where $\FAA \cap \FBB = \emptyset$ \vspace{5pt} \\
Seriality ($\mathtt{A}\text{-}\mathtt{Ser}$): & $\Fclo{\FAA} \top$ \vspace{5pt} \\
Determinism ($\mathtt{A}\text{-}\mathtt{Det}$): & $\Fclo{\FAA} (\phi \lor \psi) \rightarrow (\Fclo{\FAA} \phi \lor \Fclo{\FAG} \psi)$
\end{tabular}

\subsection{$\FMCL$ seems a minimal coalition logic for reasoning about coalitional powers}

We closely look at the alternative definition of general concurrent game models, given in Fact \ref{definition:alternative definition of general concurrent game models}. It seems they are maximally general, given their intuitive understanding.

The constraint on the set of outcome functions indicates that outcome states of a joint action of a coalition $\FAA$ are determined by all action profiles extending $\sigma_\FAA$. This reflects the principle that the future is determined by all agents. It seems hard to relax this constraint.

The constraint on the set of availability functions exactly expresses our understanding of available joint actions: a joint action is available if and only if it has an outcome. As mentioned, only the direction from left to right matters for the logic. It seems hard to think that some available action does not have an outcome.

We closely look at the axiomatic system for $\FMCL$, except for the part of propositional logic. It seems that it cannot be easily weakened, given that it is for reasoning about coalitional powers.

The axiom of monotonicity of goals says: \emph{given that $\phi$ necessarily implies $\psi$, if a coalition can ensure $\phi$, then it can ensure $\psi$.}
This sounds very plausible.

The axiom of monotonicity of coalition says: \emph{whatever a coalition can achieve can be achieved by bigger coalitions.}
In some situations, more people seem to mean less abilities, as the English proverb ``\emph{too many cooks spoil the broth}'' indicates.
We want to point out that in these situations, abilities are not treated as \emph{ontic}.


The axiom of liveness indicates: \emph{no available joint action has no outcome}.
As mentioned, this is very plausible.

The rule of conditional necessitation \emph{should} hold. Here is why. Assume it does not hold. Then there are two formulas $\phi$ and $\psi$ and a coalition $\FAA$ meeting the following conditions: $\phi$ can never be false; in a situation, $\FAA$ has an available joint action to ensure $\psi$, but $\phi$ might be false in the next moment. This seems very strange.

\subsection{Can $\FCL$ simulate lack of seriality, independence of agents, or determinism?}

Previously, we mentioned that some situations, where seriality, independence of agents, or determinism does not hold, might be indirectly representable by pointed concurrent game models.
One might wonder whether $\FCL$ can simulate the lack of the three properties.
This is an involved issue, and a thorough investigation of it would deviate too much. Here we have a brief discussion about it.

To simulate the lack of seriality in $\FCL$, in the language, we can introduce a propostional constant $\mathtt{tm}$ saying \emph{this is a terminating state};
in the semantics, we can consider only concurrent game models meeting the condition that if $\mathtt{tm}$ is true at a state, then the state has only one successor, that is, itself.
Then, every non-serial model can be transformed into a serial one.
Intuitively, $\Fclo{\FAA} \phi$ in $\FMCL$ is equivalent to $\neg \mathtt{tm} \land \Fclo{\FAA} \phi$ in the variant of $\FCL$.
Note that the variant contains new valid formulas. For example, $(\phi \land \mathtt{tm}) \rightarrow \Fclo{\FAA} \phi$ is valid. This means we get a new logic.

Can we simulate the lack of independence of agents in $\FCL$? We do not think of a natural way.

To simulate the lack of determinism in $\FCL$, we can introduce a new agent $e$ representing the environment. Intuitively, $\Fclo{\FAA} \phi$ in $\FMCL$ is equivalent to $\Fclo{\FAA} \phi$ in the variant of $\FCL$.
This variant is essentially $\FCL$, although strictly speaking, it is a proper extension of $\FCL$.

Our tentative conclusion is this: maybe $\FCL$ can simulate the lack of the three properties, but it would not be trivial.

\section{Completeness of $\FMCL$}
\label{section:Completeness of MCL}

\subsection{Our approach}

In this section, we show the completeness of $\FMCL$ in four steps.

First, we show that every formula can be transformed to an equivalent formula in a normal form, that is a conjunction of some so-called \Fdefs{standard formulas}, which are disjunctions meeting some conditions.

Second, we show a \Fdefs{downward validity lemma} for standard formulas: \textit{for every standard formula $\phi$, if $\phi$ is valid, then a \Fdefs{validity-reduction-condition} of $\phi$ is met.} Here, the validity-reduction-condition of $\phi$ concerns the validity of some formulas with lower modal depth than $\phi$.
This step is crucial for the whole proof.
In this step, we use two important notions: \Fdefs{general game forms} and \Fdefs{grafted pointed general concurrent game models}.
 
Third, we show a \Fdefs{upward derivability lemma} for standard formulas: \textit{for every standard formula $\phi$, if a \Fdefs{derivability-reduction-condition} of $\phi$ is met, then $\phi$ is derivable.} Here, the derivability-reduction-condition of $\phi$ is the result of replacing \emph{validity} in the validity-reduction-condition of $\phi$ by \emph{derivability}.

Fourth, we show by induction that \emph{for every formula $\phi$, if $\phi$ is valid, then it is derivable}. The downward validity lemma and upward derivability lemma will let us go through the induction.

\subsection{Standard formulas and a normal form lemma}

A formula $\gamma$ is called an \Fdefs{elementary disjunction} if it is a disjunction of some literals in the classical propositional logic, and $\gamma$ is called an \Fdefs{elementary conjunction} if it is a conjunction of some literals in the classical propositional logic.

For every natural number $n$, $\FNI = \{x \in \mathbb{Z} \mid -n \leq x \leq -1\}$ is called a \Fdefs{set of negative indices}, and $\FPI = \{x \in \mathbb{Z} \mid 1 \leq x \leq n\}$ is called a \Fdefs{set of positive indices}.

\begin{definition}[Standard formulas]
\label{definition:Standard formulas}

Let $\gamma$ be an elementary disjunction, $\FNI$ be a set of negative indices and $\FPI$ be a nonempty set of positive indices.
A formula $\phi$ in $\Phi_{\FMCL}$ in the form of $\gamma \lor (\FBW_{i \in \FNI} \Fclo{\FAA_i} \phi_i \rightarrow \FBV_{j \in \FPI} \Fclo{\FBB_j} \psi_j)$ is called a \Fdefs{standard formula} with respect to $\gamma$, $\FNI$ and $\FPI$, if (1) if $\FNI \neq \emptyset$, then $\Fclo{\FAA_i} \phi_i = \Fclo{\emptyset} \top$ for some $i \in \FNI$, and (2) $\Fclo{\FBB_j} \psi_j = \Fclo{\FAG} \bot$ for some $j \in \FPI$.

\end{definition}

Note $\gamma \lor (\FBW_{i \in \FNI} \Fclo{\FAA_i} \phi_i \rightarrow \FBV_{j \in \FPI} \Fclo{\FBB_j} \psi_j)$ is equivalent to $\gamma \lor \FBV_{i \in \FNI} \neg \Fclo{\FAA_i} \phi_i \lor \FBV_{j \in \FPI} \Fclo{\FBB_j} \psi_j$, from where we can see why the elements of $\FNI$ are called \emph{negative} indices and the elements of $\FPI$ are called \emph{positive} indices.

\begin{example}[Standard formulas]
\label{example:Standard formulas}

Assume $\FAG = \{a,b\}$.
Let $\gamma = \bot$, $\FNI = \{-1, -2, -3\}$ and $\FPI = \{1,2,3\}$.
Then, $\bot \lor \big(
\big( \Fclo{\FAA_{-1}} \phi_{-1} \land \Fclo{\FAA_{-2}} \phi_{-2} \land \Fclo{\FAA_{-3}} \phi_{-3} \big) \rightarrow \big( \Fclo{\FBB_{1}} \psi_1 \lor \Fclo{\FBB_{2}} \psi_2 \lor \Fclo{\FBB_{3}} \psi_3 \big)
\big)$ is a standard formula with respect to $\gamma$, $\FNI$ and $\FPI$, where
$\FAA_{-1} = \{a\}$, $\phi_{-1} = p$, $\FAA_{-2} = \{b\}$, $\phi_{-2} = q$, $\FAA_{-3} = \emptyset$, $\phi_{-3} = \top$, $\FBB_{1} = \FAG$, $\psi_1 = p \land q$, $\FBB_{2} = \{a\}$, $\psi_2 = \neg p \lor q$, $\FBB_{3} = \FAG$, $\psi_3 = \bot$.

\end{example}

Let $\gamma \lor (\FBW_{i \in \FNI} \Fclo{\FAA_i} \phi_i \rightarrow \FBV_{j \in \FPI} \Fclo{\FBB_j} \psi_j)$ be a standard formula. Define:
\[\FNI_0 = \{i \in \FNI \mid \FAA_i = \emptyset\}\]
\[\phi_{\FNI_0} = \FBW \{\phi_i \mid i \in \FNI_0\}\]

Define the \Fdefs{modal depth} of formulas in $\Phi_\FMCL$ as usual.

\begin{lemma}[Normal form]
\label{lemma:normal-form}

For every $\phi \in \Phi_\FMCL$, there is $\phi'$ such that (1) $\vdash_\FMCL \phi \leftrightarrow \phi'$, (2) $\phi$ and $\phi'$ have the same modal depth, and (3) $\phi'$ is in the form of $\chi_0 \land \dots \land \chi_k$, where every $\chi_i$ is a standard formula.

\end{lemma}

This lemma follows from the following facts: $\FMCL$ is an extension of the classical propositional logic; $\neg \bot$ and $\neg \Fclo{\FAG} \bot$ are axioms of $\FMCL$; $\Fclo{\FAA} \phi \rightarrow \Fclo{\emptyset} \top$ is derivable in $\FMCL$.

\newcommand{\mij}{i\text{-}j}

\newcommand{\moneone}{-1\text{-}1}
\newcommand{\monetwo}{-1\text{-}2}
\newcommand{\monethree}{-1\text{-}3}

\newcommand{\mtwoone}{-2\text{-}1}
\newcommand{\mtwotwo}{-2\text{-}2}
\newcommand{\mtwothree}{-2\text{-}3}

\newcommand{\mthreeone}{-3\text{-}1}
\newcommand{\mthreetwo}{-3\text{-}2}
\newcommand{\mthreethree}{-3\text{-}3}

\subsection{Downward validity lemma of standard formulas}

\begin{definition}[General game forms]
\label{definition:General game forms}
A \Fdefs{general game form} is a tuple $\FF = (s_0, S, \FAC_0, \{\Fav^0_\FAA \mid \FAA \subseteq \FAG \}, \{\Fout^0_\FAA \mid \FAA \subseteq \FAG \})$, where:
\begin{itemize}

\item

$s_0$ is a state and $S$ is a nonempty set of states.

\item

$\FAC_0$ is a nonempty set of actions.

\item

the set $\{\Fav^0_\FAA \mid \FAA \subseteq \FAG \}$ meets the following conditions, where $\FJA^0_\FAA$ is the set of joint actions of $\FAA$ with respect to $\FAC_0$ for every $\FAA \subseteq \FAG$:
\begin{itemize}

\item

for every $\FAA \subseteq \FAG$, $\Fav^0_\FAA: \{s_0\} \rightarrow \mathcal{P}(\FJA^0_\FAA)$ is an \Fdefs{availability function} for $\FAA$ at $s_0$;

\item 

for every $\FAA \subseteq \FAG$, $\Fav^0_\FAA (s_0) = \Fav_\FAG (s_0)|_\FAA$.

\end{itemize}

\emph{We notice the domain of $\Fav^0_\FAA$ is a singleton. We do this to ease some things later.}

\item

the set $\{\Fout^0_\FAA \mid \FAA \subseteq \FAG \}$ meets the following conditions:
\begin{itemize}
    
\item

for every $\FAA \subseteq \FAG$, $\Fout^0_\FAA: \{s_0\} \times \FJA^0_\FAA \rightarrow \mathcal{P} (S)$ is an \Fdefs{outcome function} for $\FAA$ at $s_0$;

\item 

for every $\sigma_\FAG \in \FJA^0_\FAG$, $\sigma_\FAG \in \Fav^0_\FAG (s_0)$ if and only if $\Fout^0_\FAG (s_0, \sigma_\FAG) \neq \emptyset$;

\item 

for every $\FAA \subseteq \FAG$ and $\sigma_\FAA \in \FJA^0_\FAA$, $\Fout^0_\FAA (s_0, \sigma_\FAA) = \bigcup \{\Fout^0_\FAG (s_0, \sigma_\FAG) \mid \sigma_\FAG \in \FJA^0_\FAG \text{ and } \sigma_\FAA$ $\subseteq \sigma_\FAG\}$.

\end{itemize}

\end{itemize}

\end{definition}

Intuitively, a general game form is an \emph{one-round game} at a state: it specifies which joint actions of a coalition are available at the state and what are the outcomes of an available joint action at the state.

\begin{definition}[Grafted pointed general concurrent game models]
\label{definition:Grafted pointed general concurrent game models}
Let $I$ be a {nonempty} set of indices, and $\{(\MM_i,s_i) \mid i \in I\}$, where $\MM_i = (\FST_i, \FAC_i, \{\Fav^i_\FAA \mid \FAA \subseteq \FAG\}, \{\Fout^i_\FAA \mid \FAA \subseteq \FAG\}, \Flab_i)$, be a set of pointed general concurrent game models meeting the following conditions: all $\FST_i$ are pairwise disjoint, and all $\FAC_i$ are pairwise disjoint. For every $i \in I$ and $\FAA \subseteq \FAG$, we use $\FJA^i_\FAA$ to indicate the set of joint actions of $\FAA$ with respect to $\FAC_i$.

Let $\gamma$ be a satisfiable elementary conjunction.

Let $\FF = (s_0, S, \FAC_0, \{\Fav^0_\FAA \mid \FAA \subseteq \FAG \}, \{\Fout^0_\FAA \mid \FAA \subseteq \FAG \})$ be a general game form where: (1) $s_0$ is a state not in any $\MM_i$; (2) $S = \{s_i \mid i \in I\}$; (3) for every $i \in I$, $\FAC_0$ and $\FAC_i$ are disjoint. For every $\FAA \subseteq \FAG$, we use $\FJA^0_\FAA$ to indicate the set of joint actions of $\FAA$ with respect to $\FAC_0$.

Define a general concurrent game model $\MM = (\FST, \FAC, \{\Fav_\FAA \mid \FAA \subseteq \FAG\}, \{\Fout_\FAA \mid \FAA \subseteq \FAG\}, \Flab)$ as follows:
\begin{itemize}

\item

$\FST =\{s_0\} \cup \bigcup \{\FST_i \mid i \in I\}$.

\item

$\FAC =\FAC_{0} \cup \bigcup \{\FAC_i \mid i \in I\}$.

\item 

for every $\FAA \subseteq \FAG$ and $s \in \FST$:
\[
\Fav_\FAA (s) =
\begin{cases}
\Fav^i_\FAA (s) & \text{if } s \in \FST_i \text{ for some } i \in I \\
\Fav^0_\FAA (s_0) & \text{if } s = s_0
\end{cases}
\]

\item 

for every $\FAA \subseteq \FAG$, $s \in \FST$ and $\sigma_\FAA \in \FJA_\FAA$:
\[
\Fout_\FAA (s, \ja{\FAA}) =
\begin{cases}
\Fout^i_\FAA (s, \sigma_\FAA) & \text{if } s \in \FST_i \text{ and } \sigma_\FAA \in \FJA^i_\FAA \text{ for some } i \in I \\
\Fout^0_\FAA (s_0, \sigma_\FAA) & \text{if } s = s_0 \text{ and } \sigma_\FAA \in \FJA^0_\FAA \\
\emptyset & \text{otherwise}
\end{cases}
\]

\item 

for every $s \in \FST$:
\[
\Flab (s) =
\begin{cases}
\Flab_i (s) & \text{if } s \in \FST_i \text{ for some } i \in I \\
\{p \mid \text{$p$ is a conjunct of $\gamma$}\} & \text{if } s = s_0 \\
\end{cases}
\]

\end{itemize}

\noindent We call $(\MM, s_0)$ a pointed general concurrent game model \Fdefs{grafted} from $\{(\MM_i,s_i) \mid i \in I\}$, $\gamma$, and $\FF$.

\end{definition}

It is easy to check that $\MM$ is a general concurrent game model. Thus, the notion of \emph{grafted pointed general concurrent game models} is well-defined.

It can be easily shown that for every $i \in I$, the \emph{generated sub-model} of $\MM_i$ at $s_i$ is also the \emph{generated sub-model} of $\MM$ at $s_i$. Then, we can get for every $i \in I$, for every $\phi \in \Phi_\FMCL$, $\MM_i,s_i \Vdash \phi$ if and only if $\MM,s_i \Vdash \phi$.

\begin{example}[A grafted pointed general concurrent game model]

Assume $\FAG = \{a,b\}$.
Let $(\MM_1,s_1)$, where $\MM_1 = (\FST_1, \FAC_1, \{\Fav^1_\FAA \mid \FAA \subseteq \FAG\}, \{\Fout^1_\FAA \mid \FAA \subseteq \FAG\}, \Flab_1)$, be a pointed general concurrent game model, indicated by Figure \ref{figure:a pointed general concurrent game model (M1,s1)},
and $(\MM_2,s_2)$, where $\MM_2 = (\FST_2, \FAC_2, \{\Fav^2_\FAA \mid \FAA \subseteq \FAG\}, \{\Fout^2_\FAA \mid \FAA \subseteq \FAG\}, \Flab_2)$, be a pointed general concurrent game model, indicated by Figure \ref{figure:a pointed general concurrent game model (M2,s2)}.

Let $\gamma = \neg p \land q \land r$.

Let $\FF = (s_0, S, \FAC_0, \{\Fav^0_\FAA \mid \FAA \subseteq \FAG \}, \{\Fout^0_\FAA \mid \FAA \subseteq \FAG \})$ be a general game form, as indicated in Figure \ref{figure:a general game form}.

Figure \ref{figure:a grafted pointed general concurrent game model} illustrates the pointed general concurrent game model grafted from $\{(\MM_1,s_1), (\MM_2,s_2)\}$, $\gamma$ and $\FF$.

\end{example}


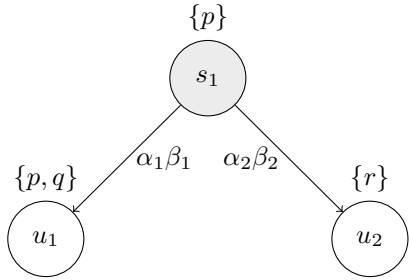
\begin{figure}
\begin{center}
\begin{tikzpicture}
[
->=stealth,
scale=1,
every node/.style={transform shape},
]
\tikzstyle{every state}=[minimum size=10mm]

\node[state,fill=gray!15] (s-1) {$s_1$};

\node[state,position=225:{20mm} from s-1] (s-1-1) {$u_1$};
\node[state,position=315:{20mm} from s-1] (s-1-n) {$u_2$};

\node[above=5mm] (a-s-1) at (s-1) {$\{p\}$};

\node[above=5mm] (a-s-1-1) at (s-1-1) {$\{p,q\}$};
\node[above=5mm] (a-s-1-n) at (s-1-n) {$\{r\}$};

\path

(s-1) edge [right] node {$\alpha_1 \beta_1$} (s-1-1)
(s-1) edge [left] node {$\alpha_2 \beta_2$} (s-1-n)
;

\end{tikzpicture}

\caption{This figure indicates a pointed general concurrent game model $(\MM_1,s_1)$.
}
\label{figure:a pointed general concurrent game model (M1,s1)}

\end{center}

\end{figure}


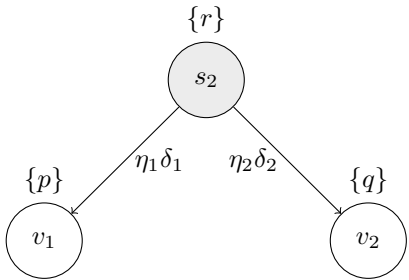
\begin{figure}
\begin{center}
\begin{tikzpicture}
[
->=stealth,
scale=1,
every node/.style={transform shape},
]
\tikzstyle{every state}=[minimum size=10mm]

\node[state,fill=gray!15] (s-n) {$s_2$};

\node[state,position=225:{20mm} from s-n] (s-n-1) {$v_1$};
\node[state,position=315:{20mm} from s-n] (s-n-n) {$v_2$};

\node[above=5mm] (a-s-n) at (s-n) {$\{r\}$};

\node[above=5mm] (a-s-n-1) at (s-n-1) {$\{p\}$};
\node[above=5mm] (a-s-n-n) at (s-n-n) {$\{q\}$};

\path
(s-n) edge [right] node {$\eta_1 \delta_1$} (s-n-1)
(s-n) edge [left] node {$\eta_2 \delta_2$} (s-n-n)
;

\end{tikzpicture}

\caption{This figure indicates a pointed general concurrent game model $(\MM_2,s_2)$.
}

\label{figure:a pointed general concurrent game model (M2,s2)}

\end{center}

\end{figure}


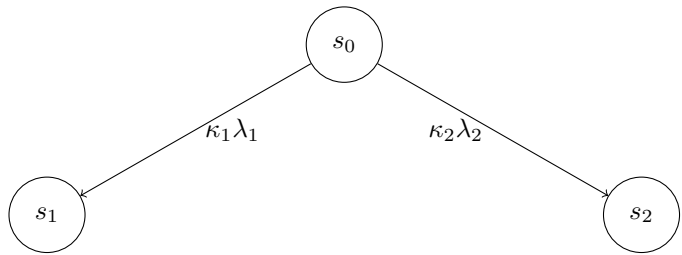
\begin{figure}
\begin{center}
\begin{tikzpicture}
[
->=stealth,
scale=1,
every node/.style={transform shape},
]
\tikzstyle{every state}=[minimum size=10mm]

\node[state] (s) {$s_0$};

\node[state,position=210:{35mm} from s] (s-1) {$s_1$};
\node[state,position=330:{35mm} from s] (s-n) {$s_2$};

\path
(s) edge [right] node {$\kappa_1 \lambda_1$} (s-1)
(s) edge [left] node {$\kappa_2 \lambda_2$} (s-n)
;

\end{tikzpicture}

\caption{This figure indicates a general game form $\FF$.
}
\label{figure:a general game form}

\end{center}

\end{figure}


\begin{figure}
\begin{center}
\begin{tikzpicture}
[
->=stealth,
scale=1,
every node/.style={transform shape},
]
\tikzstyle{every state}=[minimum size=10mm]

\node[state,fill=gray!15] (s) {$s_0$};

\node[state,position=210:{35mm} from s] (s-1) {$s_1$};
\node[state,position=330:{35mm} from s] (s-n) {$s_2$};

\node[state,position=225:{20mm} from s-1] (s-1-1) {$u_1$};
\node[state,position=315:{20mm} from s-1] (s-1-n) {$u_2$};

\node[state,position=225:{20mm} from s-n] (s-n-1) {$v_1$};
\node[state,position=315:{20mm} from s-n] (s-n-n) {$v_2$};

\node[above=5mm] (a-s) at (s) {$\{q,r\}$};

\node[above=5mm] (a-s-1) at (s-1) {$\{p\}$};
\node[above=5mm] (a-s-n) at (s-n) {$\{r\}$};

\node[above=5mm] (a-s-1-1) at (s-1-1) {$\{p,q\}$};
\node[above=5mm] (a-s-1-n) at (s-1-n) {$\{r\}$};

\node[above=5mm] (a-s-n-1) at (s-n-1) {$\{p\}$};
\node[above=5mm] (a-s-n-n) at (s-n-n) {$\{q\}$};

\path
(s) edge [right] node {$\kappa_1 \lambda_1$} (s-1)
(s) edge [left] node {$\kappa_2 \lambda_2$} (s-n)

(s-1) edge [right] node {$\alpha_1 \beta_1$} (s-1-1)
(s-1) edge [left] node {$\alpha_2 \beta_2$} (s-1-n)

(s-n) edge [right] node {$\eta_1 \delta_1$} (s-n-1)
(s-n) edge [left] node {$\eta_2 \delta_2$} (s-n-n)
;

\end{tikzpicture}

\caption{This figure indicates a pointed general concurrent game model $(\MM,s_0)$, grafted from: a pointed general concurrent game model $(\MM_1,s_1)$, where $\MM_1 = (\FST_1, \FAC_1, \{\Fav^1_\FAA \mid \FAA \subseteq \FAG\}, \{\Fout^1_\FAA \mid \FAA \subseteq \FAG\}, \Flab_1)$, indicated by Figures \ref{figure:a pointed general concurrent game model (M1,s1)}; a pointed general concurrent game model $(\MM_2,s_2)$, where $\MM_2 = (\FST_2, \FAC_2, \{\Fav^2_\FAA \mid \FAA \subseteq \FAG\}, \{\Fout^2_\FAA \mid \FAA \subseteq \FAG\}, \Flab_2)$, indicated by Figure \ref{figure:a pointed general concurrent game model (M2,s2)}; the elementary conjunction $\neg p \land q \land r$; a general game form $\FF = (s_0, S, \FAC_0, \{\Fav^0_\FAA \mid \FAA \subseteq \FAG \}, \{\Fout^0_\FAA \mid \FAA \subseteq \FAG \})$, indicated by Figure \ref{figure:a general game form}.
At $(\MM,s_0)$:
(1) available joint actions of a coalition $\FAA$ at $s_0$ is specified by $\Fav^0_\FAA$, and available joint actions of $\FAA$ at other states are as in the corresponding models;
(2) outcome states of an available joint action $\sigma_\FAA$ of a coalition $\FAA$ at $s_0$ is specified by $\Fout^0_\FAA$, and outcome states of $\sigma_\FAA$ at other states are as in the corresponding models;
(3) at $s_0$, only $q$ and $r$ are true, which are conjuncts of $\neg p \land q \land r$, and truth values of atomic propositions at other states are as in the corresponding models.
}
\label{figure:a grafted pointed general concurrent game model}

\end{center}

\end{figure}
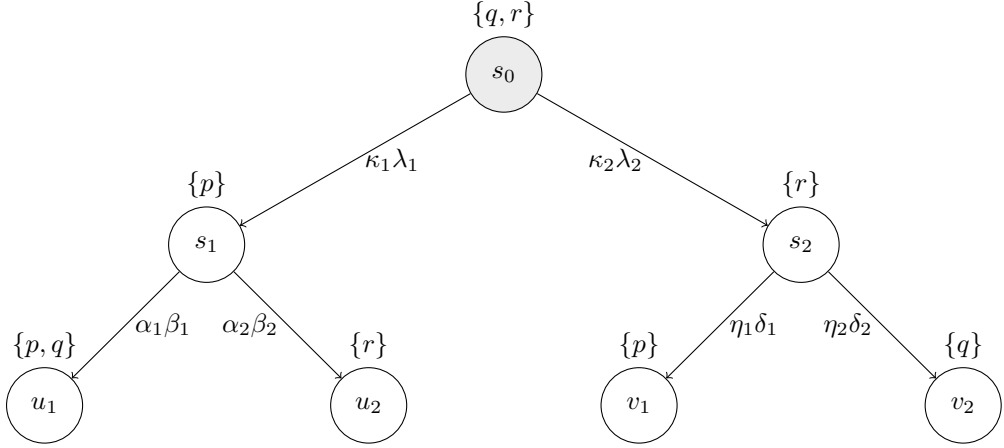

\begin{lemma}[Downward validity of standard formulas]
\label{lemma:Downward validity of standard formulas}
Let $\gamma \lor (\FBW_{i \in \FNI} \Fclo{\FAA_i} \phi_i \rightarrow \FBV_{j \in \FPI} \Fclo{\FBB_j} \psi_j)$ be a standard formula.

Assume $\vDash \gamma \lor (\FBW_{i \in \FNI} \Fclo{\FAA_i} \phi_i \rightarrow \FBV_{j \in \FPI} \Fclo{\FBB_j} \psi_j)$.

Then, one of the following two conditions is met:
\begin{enumerate}[label=(\alph*),leftmargin=3.33em]

\item 

$\vDash \gamma$;

\item 

there is $i \in \FNI$ and $j \in \FPI$ such that $\FAA_i \subseteq \FBB_{j}$ and $\vDash (\phi_{\FNI_0} \land \phi_{i}) \rightarrow \psi_{j}$.

\end{enumerate}

\end{lemma}

\begin{example}
\label{example:11}

Consider the standard formula 
$\phi = \gamma \lor \big(
\big( \Fclo{\FAA_{-1}} \phi_{-1} \land \Fclo{\FAA_{-2}} \phi_{-2} \land \Fclo{\FAA_{-3}} \phi_{-3} \big) \rightarrow \big( \Fclo{\FBB_{1}} \psi_1 \lor \Fclo{\FBB_{2}} \psi_2 \lor \Fclo{\FBB_{3}} \psi_3 \big)
\big)$ with respect to $\gamma = \bot$, $\FNI = \{-1,-2,-3\}$ and $\FPI = \{1,2,3\}$, given in Example \ref{example:Standard formulas}, where $\FAA_{-1} = \{a\}$, $\phi_{-1} = p$, $\FAA_{-2} = \{b\}$, $\phi_{-2} = q$, $\FAA_{-3} = \emptyset$, $\phi_{-3} = \top$, $\FBB_{1} = \FAG$, $\psi_1 = p \land q$, $\FBB_{2} = \{a\}$, $\psi_2 = \neg p \lor q$, $\FBB_{3} = \FAG$, $\psi_3 = \bot$.
Note $\FNI_0 = \{-3\}$ and $\phi_{\FNI_0} = \top$.

By this lemma, if $\models \phi$, then one of the following holds:
\begin{enumerate}[label=(\arabic*),leftmargin=3.33em]

\item 

$\models \gamma$, that is, $\models \bot$;

\item 

$\FAA_{-1} \subseteq \FBB_1$ and $\models (\phi_{\FNI_0} \land \phi_{-1}) \rightarrow \psi_1$, that is, $\{a\} \subseteq \FAG$ and $\models (\top \land p) \rightarrow (p \land q)$;

\item 

$\FAA_{-1} \subseteq \FBB_2$ and $\models (\phi_{\FNI_0} \land \phi_{-1}) \rightarrow \psi_2$, that is, $\{a\} \subseteq \{a\}$ and $\models (\top \land p) \rightarrow (\neg p \lor q)$;

\item 

$\FAA_{-1} \subseteq \FBB_3$ and $\models (\phi_{\FNI_0} \land \phi_{-1}) \rightarrow \psi_2$, that is, $\{a\} \subseteq \FAG$ and $\models (\top \land p) \rightarrow \bot$;

\item 

$\FAA_{-2} \subseteq \FBB_1$ and $\models (\phi_{\FNI_0} \land \phi_{-2}) \rightarrow \psi_1$, that is, $\{b\} \subseteq \FAG$ and $\models (\top \land q) \rightarrow (p \land q)$;

\item 

$\FAA_{-2} \subseteq \FBB_2$ and $\models (\phi_{\FNI_0} \land \phi_{-2}) \rightarrow \psi_2$, that is, $\{b\} \subseteq \{a\}$ and $\models (\top \land q) \rightarrow (\neg p \lor q)$;

\item 

$\FAA_{-2} \subseteq \FBB_3$ and $\models (\phi_{\FNI_0} \land \phi_{-2}) \rightarrow \psi_3$, that is, $\{b\} \subseteq \FAG$ and $\models (\top \land q) \rightarrow \bot$;

\item 

$\FAA_{-3} \subseteq \FBB_1$ and $\models (\phi_{\FNI_0} \land \phi_{-3}) \rightarrow \psi_1$, that is, $\emptyset \subseteq \FAG$ and $\models (\top \land \top) \rightarrow (p \land q)$;

\item 

$\FAA_{-3} \subseteq \FBB_2$ and $\models (\phi_{\FNI_0} \land \phi_{-3}) \rightarrow \psi_2$, that is, $\emptyset \subseteq \{a\}$ and $\models (\top \land \top) \rightarrow (\neg p \lor q)$;

\item 

$\FAA_{-3} \subseteq \FBB_3$ and $\models (\phi_{\FNI_0} \land \phi_{-3}) \rightarrow \psi_2$, that is, $\emptyset \subseteq \FAG$ and $\models (\top \land \top) \rightarrow \bot$.

\end{enumerate}

\end{example}

We now give the proof of the lemma. In presenting the proof, we offer some examples to help readers understand it.

\begin{proof}
~

Assume (a) $\neg \gamma$ is satisfiable, and (b) for all $i \in \FNI$ and $j \in \FPI$ such that $\FAA_i \subseteq \FBB_{j}$, $\phi_{\FNI_0} \land \phi_{i} \land \neg \psi_{j}$ is satisfiable. It suffices to show $\neg \gamma \land \FBW_{i \in \FNI} \Fclo{\FAA_i} \phi_i \land \FBW_{j \in \FPI} \neg \Fclo{\FBB_j} \psi_j$ is satisfiable.

\begin{example}
\label{example:12}

Consider Example \ref{example:11}.
It can be seen that the two assumptions (a) and (b) hold:
\begin{itemize}

\item 

$\neg \gamma$ is satisfiable, that is, $\neg \bot$ is satisfiable;

\item 

$\FAA_{-1} \subseteq \FBB_1$ and $\phi_{\FNI_0} \land \phi_{-1} \land \neg \psi_1$ is satisfiable, that is, $\{a\} \subseteq \FAG$ and $\top \land p \land \neg (p \land q)$ is satisfiable;

\item 

$\FAA_{-1} \subseteq \FBB_2$ and $\phi_{\FNI_0} \land \phi_{-1} \land \neg \psi_2$ is satisfiable, that is, $\{a\} \subseteq \{a\}$ and $\top \land p \land \neg (\neg p \lor q)$ is satisfiable;

\item 

$\FAA_{-1} \subseteq \FBB_3$ and $\phi_{\FNI_0} \land \phi_{-1} \land \neg \psi_3$ is satisfiable, that is, $\{a\} \subseteq \FAG$ and $\top \land p \land \neg \bot$ is satisfiable;

\item 

$\FAA_{-2} \subseteq \FBB_1$ and $\phi_{\FNI_0} \land \phi_{-2} \land \neg \psi_1$ is satisfiable, that is, $\{b\} \subseteq \FAG$ and $\top \land q \land \neg (p \land q)$ is satisfiable;

\item 

$\FAA_{-2} \subseteq \FBB_3$ and $\phi_{\FNI_0} \land \phi_{-2} \land \neg \psi_3$ is satisfiable, that is, $\{b\} \subseteq \FAG$ and $\top \land q \land \neg \bot$ is satisfiable;

\item 

$\FAA_{-3} \subseteq \FBB_1$ and $\phi_{\FNI_0} \land \phi_{-3} \land \neg \psi_1$ is satisfiable, that is, $\emptyset \subseteq \FAG$ and $\top \land \top \land \neg (p \land q)$ is satisfiable;

\item 

$\FAA_{-3} \subseteq \FBB_2$ and $\phi_{\FNI_0} \land \phi_{-3} \land \neg \psi_2$ is satisfiable, that is, $\emptyset \subseteq \{a\}$ and $\top \land \top \land \neg (\neg p \lor q)$ is satisfiable;

\item 

$\FAA_{-3} \subseteq \FBB_3$ and $\phi_{\FNI_0} \land \phi_{-3} \land \neg \psi_3$ is satisfiable, that is, $\emptyset \subseteq \FAG$ and $\top \land \top \land \neg \bot$ is satisfiable.

\end{itemize}

Note $\FAA_{-2} \not \subseteq \FBB_2$. By what follows, we can show $
\neg \gamma \land
\Fclo{\FAA_{-1}} \phi_{-1} \land \Fclo{\FAA_{-2}} \phi_{-2} \land \Fclo{\FAA_{-3}} \phi_{-3} \land \neg \Fclo{\FBB_{1}} \psi_1 \land \neg \Fclo{\FBB_{2}} \psi_2 \land \neg \Fclo{\FBB_{3}} \psi_3
$
is satisfiable.

\end{example}

------------------------------------

Assume $\FNI = \emptyset$. It is easy to construct a pointed general concurrent game model $(\MM,s_0)$ meeting the following conditions:
\begin{enumerate}[label=(\arabic*),leftmargin=3.33em]

\item

$\MM,s_0 \Vdash \neg \gamma$;

\item

at $s_0$, all coalitions have no available joint actions.

\end{enumerate}

Then $\MM,s_0 \Vdash \neg \gamma \land \FBW_{j \in \FPI} \neg \Fclo{\FBB_j} \psi_j$.

------------------------------------

Assume $\FNI \neq \emptyset$. 

Then, $\Fclo{\FAA_i} \phi_i = \Fclo{\emptyset} \top$ for some $i \in \FNI$. Note $\Fclo{\FBB_j} \psi_j = \Fclo{\FAG} \bot$ for some $j \in \FPI$. Then, the following conditions hold:
\begin{enumerate}[label=(\arabic*),leftmargin=3.33em]

\item 

there is $i \in \FNI$ and $j \in \FPI$ such that $\FAA_i \subseteq \FBB_{j}$;

\item 

for every $i \in \FNI$, there is $j \in \FPI$ such that $\FAA_i \subseteq \FBB_{j}$;

\item 

for every $j \in \FPI$, there is $i \in \FNI$ such that $\FAA_i \subseteq \FBB_{j}$.

\end{enumerate}

\medskip

{
We will show $\neg \gamma \land \FBW_{i \in \FNI} \Fclo{\FAA_i} \phi_i \land \FBW_{j \in \FPI} \neg \Fclo{\FBB_j} \psi_j$ is satisfiable by the following steps.
\begin{enumerate}[leftmargin=5.33em]

\item[Step 1:]

We construct a pointed general concurrent game model $(\MM,s_0)$, which is grafted from a set of pointed general concurrent game models, an elementary conjunction, and a general game form. This step consists of three sub-steps:
\begin{enumerate}[leftmargin=4.1em]

\item[Step 1a:]

specifying a set of pointed general concurrent game models.

\item[Step 1b:]

giving an elementary conjunction.

\item[Step 1c:]

defining a general game form.

\end{enumerate}

\item[Step 2:]

We show $\neg \gamma \land \FBW_{i \in \FNI} \Fclo{\FAA_i} \phi_i \land \FBW_{j \in \FPI} \neg \Fclo{\FBB_j} \psi_j$ is true at $(\MM,s_0)$. This step consists of three sub-steps:
\begin{enumerate}[leftmargin=4.1em]

\item[Step 2a:]

showing $\neg \gamma$ is true at $(\MM,s_0)$.

\item[Step 2b:]

showing $\Fclo{\FAA_i} \phi_i$ is true at $(\MM,s_0)$ for every $i \in \FNI$.

\item[Step 2c:]

showing $\neg \Fclo{\FBB_j} \psi_j$ is true at $(\MM,s_0)$ for every $j \in \FPI$.

\end{enumerate}

\end{enumerate}
}

\medskip

\textbf{Step 1a: specifying a set of pointed general concurrent game models.}

\medskip

Let $\{(\MM_{\mij},s_{\mij}) \mid i \in \FNI, j \in \FPI, \FAA_i \subseteq \FBB_j\}$, where $\MM_{\mij} = (\FST_{\mij}, \FAC_{\mij}, \{ \Fav^{\mij}_\FAA \mid \FAA \subseteq \FAG \}, \{ \Fout^{\mij}_\FAA \mid \FAA \subseteq \FAG \}, \Flab_{\mij})$, be a set of pointed general concurrent game models meeting the following conditions:
\begin{enumerate}[label=(\arabic*),leftmargin=3.33em]

\item

for all $i \in \FNI$ and $j \in \FPI$ such that $\FAA_i \subseteq \FBB_j$, $\MM_{\mij},s_{\mij} \Vdash \phi_{\FNI_0} \land \phi_{i} \land \neg \psi_j$;

\item

all $\FST_{\mij}$ are pairwise disjoint;

\item

all $\FAC_{\mij}$ are pairwise disjoint.

\end{enumerate}

Note there is $i \in \FNI$ and $j \in \FPI$ such that $\FAA_i \subseteq \FBB_{j}$. Then the set $\{(\MM_{\mij},s_{\mij}) \mid i \in \FNI, j \in \FPI, \FAA_i \subseteq \FBB_j\}$ is not empty.

\medskip

\textbf{Step 1b: giving an elementary conjunction.}

\medskip

Let $\gamma'$ be a satisfiable elementary conjunction equivalent to $\neg \gamma$.

\medskip

\textbf{Step 1c: defining a general game form.}

\medskip

Define a general game form $\FF = (s_0, S, \FAC_0, \{\Fav^0_\FAA \mid \FAA \subseteq \FAG \}, \{\Fout^0_\FAA \mid \FAA \subseteq \FAG \})$ as follows, where we use $\FJA^0_\FAA$ to indicate the set of joint actions of $\FAA$ with respect to $\FAC_0$ for every $\FAA \subseteq \FAG$:
\begin{itemize}

\item

$s_0$ is a state not in any $\MM_{\mij}$.

\item

$S$ is the set consisting of all $s_{\mij}$.

Note that $S$ is not empty.

\item

$\FAC_0 = \{\alpha_i \mid i \in \FNI\} \cup \{\beta_{\mij} \mid i \in \FNI, j \in \FPI, \FAA_i \not \subseteq \FBB_j\}$, which is such that $\FAC_0$ and $\FAC_{\mij}$ disjoint for all $\FAC_{\mij}$.

Note $\FAC_0$ is not empty.

\item 

We specify two kinds of available action profiles.

For all $i \in \FNI$, let $\sigma_\FAG^{i}$ be a joint action of $\FAG$ such that for all $a \in \FAG$, $\sigma_\FAG^{i} (a) = \alpha_i$.
For every action profile $\sigma_\FAG^{i}$, all agents make the same choice, that is, $\alpha_i$.

For all $i \in \FNI$ and $j \in \FPI$ such that $\FAA_i \not \subseteq \FBB_j$, let $a_{\mij} \in \FAG$ such that $a_{\mij} \in \FAA_i$ but $a_{\mij} \notin \FBB_j$.
For all $i \in \FNI$ and $j \in \FPI$ such that $\FAA_i \not \subseteq \FBB_j$, let $\Fspac^{\mij}_\FAG$ be a joint action of $\FAG$ such that for all $a \in \FAG$:
\[
\Fspac^{\mij}_\FAG (a) =
\begin{cases}
\beta_{\mij} & \text{if } a = a_{\mij} \\
\sigma_\FAG^{i} (a) & \text{otherwise}
\end{cases}
\]
For every action profile $\Fspac^{\mij}_\FAG$, the agent $a_{\mij}$ makes the choice $\beta_{\mij}$, and all other agents make the same choice, that is, $\alpha_i$.

\item

For all $\FAA \subseteq \FAG$:
\begin{itemize}

\item 

if $\FAA = \FAG$, then $\Fav^0_\FAA (s_0) = \{\sigma_\FAG^{i} \mid i \in \FNI\} \cup \{\Fspac^{\mij}_\FAG \mid i \in \FNI, j \in \FPI, \FAA_i \not \subseteq \FBB_j\}$;

\item 

if $\FAA \neq \FAG$, then $\Fav^0_\FAA (s_0) = \Fav^0_\FAG (s_0)|_\FAA$.

\end{itemize}

\item

For all $\FAA \subseteq \FAG$ and $\sigma_\FAA \in \FJA^0_\FAA$:
\begin{itemize}

\item 

if $\FAA = \FAG$ and $\sigma_\FAA \in \Fav^0_\FAA (s_0)$, then:
\begin{itemize}

\item 

if $\sigma_\FAA = \sigma_\FAG^{i}$ for some $i \in \FNI$, then $\Fout^0_\FAA (s_0, \sigma_\FAA) = \{s_{\mij} \in S \mid j \in \FPI\}$;

\item 

if $\sigma_\FAA = \Fspac^{\mij}_\FAG$ for some $i \in \FNI$ and $j \in \FPI$ such that $\FAA_i \not \subseteq \FBB_j$, then $\Fout^0_\FAA (s_0, \sigma_\FAA) = S$;

\end{itemize}

\item 

if $\FAA = \FAG$ and $\sigma_\FAA \notin \Fav^0_\FAA (s)$, then $\Fout^0_\FAA (s_0, \sigma_\FAA) = \emptyset$;

\item 

if $\FAA \neq \FAG$, then $\Fout^0_\FAA (s_0, \sigma_\FAA) = \bigcup \{\Fout^0_\FAG (s_0, \sigma_\FAG) \mid \sigma_\FAG \in \FJA^0_\FAG \text{ and } \sigma_\FAA \subseteq \sigma_\FAG\}$.

\end{itemize}

\end{itemize}

Let $(\MM,s_0)$, where $\MM = (\FST, \FAC, \{\Fav_\FAA \mid \FAA \subseteq \FAG\}, \{\Fout_\FAA \mid \FAA \subseteq \FAG\}, \Flab)$, be a pointed general concurrent game model grafted from $\{(\MM_{\mij},s_{\mij}) \mid i \in \FNI, j \in \FPI, \FAA_i \subseteq \FBB_j\}$, $\gamma'$, and $\FF$.
We claim $\MM,s_0 \Vdash \neg \gamma \land \FBW_{i \in \FNI} \Fclo{\FAA_i} \phi_i \land \FBW_{j \in \FPI} \neg \Fclo{\FBB_j} \psi_j$.

\begin{example}
\label{example:13}

Consider Example \ref{example:12}.
Let $(\MM_{\moneone},s_{\moneone})$, 
$(\MM_{\monetwo},s_{\monetwo})$, 
$(\MM_{\monethree},$ $s_{\monethree})$, $(\MM_{\mtwoone},$ $s_{\mtwoone})$,
$(\MM_{\mtwothree}, s_{\mtwothree})$, $(\MM_{\mthreeone},s_{\mthreeone})$, 
$(\MM_{\mthreetwo},s_{\mthreetwo})$, $(\MM_{\mthreethree},s_{\mthreethree})$ be eight pointed general concurrent game models such that their domains of states are pairwise disjoint, their domains of actions are pairwise disjoint, and they respectively satisfy the eight formulas mentioned in Example \ref{example:12}:
\begin{itemize}

\item 

$\MM_{\moneone},s_{\moneone} \Vdash \top \land p \land \neg (p \land q)$;

\item 

$\MM_{\monetwo},s_{\monetwo} \Vdash \top \land p \land \neg (\neg p \lor q)$;

\item 

$\MM_{\monethree},s_{\monethree} \Vdash \top \land p \land \neg \bot$;

\item 

$\MM_{\mtwoone},s_{\mtwoone} \Vdash \top \land q \land \neg (p \land q)$;

\item 

$\MM_{\mtwothree},s_{\mtwothree} \Vdash \top \land q \land \neg \bot$;

\item 

$\MM_{\mthreeone},s_{\mthreeone} \Vdash \top \land \top \land \neg (p \land q)$;

\item 

$\MM_{\mthreetwo},s_{\mthreetwo} \Vdash \top \land \top \land \neg (\neg p \lor q)$;

\item 

$\MM_{\mthreethree},s_{\mthreethree} \Vdash \top \land \top \land \neg \bot$.

\end{itemize}

Let $\gamma' = \top$.

Define a general game form $\FF = (s_0, S, \FAC_0, \{\Fav^0_\FAA \mid \FAA \subseteq \FAG \}, \{\Fout^0_\FAA \mid \FAA \subseteq \FAG \})$ as follows:
\begin{itemize}

\item

$S = \{
s_{\moneone}, s_{\monetwo}, s_{\monethree}, s_{\mtwoone}, s_{\mtwothree}, s_{\mthreeone}, s_{\mthreetwo}, s_{\mthreethree}
\}$;

\item

$\FAC_0 = \{\alpha_{-1}, \alpha_{-2}, \alpha_{-3}\} \cup \{\beta_{\mtwotwo}\}$;

\item 

Note $\FAA_{-2} \subseteq \FBB_2$. Let $a_{\mtwotwo} = b$. Let $\sigma_\FAG^{-1} = \alpha_{-1} \alpha_{-1}$, $\sigma_\FAG^{-2} = \alpha_{-2} \alpha_{-2}$, $\sigma_\FAG^{-3} = \alpha_{-3} \alpha_{-3}$, $\lambda^{\mtwotwo}_\FAG = \alpha_{-2} \beta_{\mtwotwo}$.

$\Fav^0_\FAG (s_0) = \{\sigma_\FAG^{-1}, \sigma_\FAG^{-2}, \sigma_\FAG^{-3}, \lambda^{\mtwotwo}_\FAG\}$;

$\Fav^0_a (s_0) = \{\alpha_{-1}, \alpha_{-2}, \alpha_{-3}\}$;

$\Fav^0_b (s_0) = \{\alpha_{-1}, \alpha_{-2}, \alpha_{-3}, \beta_{\mtwotwo}\}$;

$\Fav^0_{\emptyset} (s_0) = \{\emptyset\}$;

\item

$\Fout^0_\FAG (s_0, \sigma_\FAG^{-1}) = \{s_{\moneone}, s_{\monetwo}, s_{\monethree}\}$;

$\Fout^0_\FAG (s_0, \sigma_\FAG^{-2}) = \{s_{\mtwoone}, s_{\mtwothree}\}$;

$\Fout^0_\FAG (s_0, \sigma_\FAG^{-3}) = \{s_{\mthreeone}, s_{\mthreetwo}, s_{\mthreethree}\}$;

$\Fout^0_\FAG (s_0, \lambda^{\mtwotwo}_\FAG) = S$;

$\Fout^0_a (s_0, \alpha_{-1}) = \{s_{\moneone}, s_{\monetwo}, s_{\monethree}\}$;

$\Fout^0_a (s_0, \alpha_{-2}) = S$;

$\Fout^0_a (s_0, \alpha_{-3}) = \{s_{\mthreeone}, s_{\mthreetwo}, s_{\mthreethree}\}$;

$\Fout^0_b (s_0, \alpha_{-1}) = \{s_{\moneone}, s_{\monetwo}, s_{\monethree}\}$;

$\Fout^0_b (s_0, \alpha_{-2}) = \{s_{\mtwoone}, s_{\mtwothree}\}$;

$\Fout^0_b (s_0, \alpha_{-3}) = \{s_{\mthreeone}, s_{\mthreetwo}  s_{\mthreethree}\}$;

$\Fout^0_b (s_0, \beta_{\mtwotwo}) = S$;

$\Fout^0_\emptyset (s_0, \emptyset) = S$.

\end{itemize}

Let $(\MM,s_0)$, where $\MM = (\FST, \FAC, \{\Fav_\FAA \mid \FAA \subseteq \FAG\}, \{\Fout_\FAA \mid \FAA \subseteq \FAG\}, \Flab)$, be a pointed general concurrent game model grafted from $\{
(\MM_{\moneone},s_{\moneone}), (\MM_{\monetwo},s_{\monetwo}),(\MM_{\monethree},s_{\monethree}), (\MM_{\mtwoone},s_{\mtwoone}),$ $(\MM_{\mtwothree},$ $s_{\mtwothree}), (\MM_{\mthreeone},s_{\mthreeone}), 
(\MM_{\mthreetwo},s_{\mthreetwo}),(\MM_{\mthreethree},s_{\mthreethree})
\}$, $\gamma'$ and $\FF$.

By what follows, we can show $(\MM,s_0)$ satisfies 
$
\neg \gamma \land
\Fclo{\FAA_{-1}} \phi_{-1} \land \Fclo{\FAA_{-2}} \phi_{-2} \land \Fclo{\FAA_{-3}} \phi_{-3} \land \neg \Fclo{\FBB_{1}} \psi_1 \land \neg \Fclo{\FBB_{2}} \psi_2 \land \neg \Fclo{\FBB_{3}} \psi_3
$.

\end{example}

\textbf{Step 2a: showing $\neg \gamma$ is true at $(\MM,s_0)$.}

\medskip

From the definition of grafted pointed general concurrent game models, it is easy to see $\MM,s_0 \Vdash \gamma'$. Then, $\MM,s_0 \Vdash \neg \gamma$.

\medskip

\textbf{Step 2b: showing $\Fclo{\FAA_i} \phi_i$ is true at $(\MM,s_0)$ for every $i \in \FNI$.}

\medskip

Let $i \in \FNI$. Note $\sigma_\FAG^{i}$ is an available joint action of $\FAG$ at $s_0$. Then $\sigma_\FAG^{i}|_{\FAA_i}$ is an available joint action of $\FAA_i$ at $s_0$.

Assume $\FAA_i = \emptyset$.
Then $i \in \FNI_0$. Note for all $s \in S$, $\MM,s \Vdash \phi_{\FNI_0}$. Then for all $s \in S$, $\MM,s \Vdash \phi_i$. Then, $\MM,s_0 \Vdash \Fclo{\FAA_i} \phi_i$.

Assume $\FAA_i \neq \emptyset$.

We claim that for any available joint action $\delta_\FAG$ of $\FAG$ at $s_0$, if $\sigma_\FAG^{i}|_{\FAA_i} \subseteq \delta_\FAG$, then $\delta_\FAG = \sigma_\FAG^{i}$.

Assume not. Then at $s_0$, $\overline{\FAA_i}$ has an available joint action $\sigma_{\overline{\FAA_i}}$ such that $\sigma_\FAG^{i}|_{\FAA_i} \cup \sigma_{\overline{\FAA_i}}$ is an available action for $\FAG$ and $\sigma_{\overline{\FAA_i}} \neq \sigma_\FAG^{i}|_{\overline{\FAA_i}}$.

Let $a \in \FAA_i$. Then $\sigma_\FAG^{i}|_{\FAA_i} (a) = \alpha_i$.

Note that at $s_0$, there are only two kinds of available joint actions for $\FAG$: $\sigma^{m}_\FAG$ and $\Fspac^{m\text{-}n}_\FAG$.

Suppose $\sigma_\FAG^{i}|_{\FAA_i} \cup \sigma_{\overline{\FAA_i}} = \sigma^{i'}_\FAG$ for some $i' \in \FNI$. Then for all $x \in \FAG$, $\sigma^{i'}_\FAG (x) = \alpha_{i'}$. Note $\sigma^{i'}_\FAG (a) = \sigma_\FAG^{i}|_{\FAA_i} (a) = \alpha_i$. Then $i' = i$. Then $\sigma_\FAG^{i}|_{\FAA_i} \cup \sigma_{\overline{\FAA_i}} = \sigma^{i}_\FAG$. Then $\sigma_{\overline{\FAA_i}} = \sigma_\FAG^{i}|_{\overline{\FAA_i}}$. We have a contradiction.

Suppose $\sigma_\FAG^{i}|_{\FAA_i} \cup \sigma_{\overline{\FAA_i}} = \Fspac^{i'\text{-}j}_\FAG$ for some $i' \in \FNI$ and $j \in \FPI$. Then $\FAA_{i'} \not \subseteq \FBB_j$, $a_{i'\text{-}j} \in \FAA_{i'}$, $a_{i'\text{-}j} \notin \FBB_j$, and $\Fspac^{i'\text{-}j}_\FAG (a_{i'\text{-}j}) = \beta_{i'\text{-}j}$.

Note either $\Fspac^{i'\text{-}j}_\FAG (a) = \beta_{i'\text{-}j}$, or $\Fspac^{i'\text{-}j}_\FAG (a) = \alpha_{i'}$. Also note $\Fspac^{i'\text{-}j}_\FAG (a) = \sigma_\FAG^{i}|_{\FAA_i} (a) = \alpha_i$. Then it is impossible that $\Fspac^{i'\text{-}j}_\FAG (a) = \beta_{i'\text{-}j}$. Then $\Fspac^{i'\text{-}j}_\FAG (a) = \alpha_{i'}$. Then $i = i'$.
Then $a_{\mij} \in \FAA_{i}$ and $\Fspac^{\mij}_\FAG (a_{\mij}) = \beta_{\mij}$. Then $\Fspac^{\mij}_\FAG (a_{\mij}) = \sigma_\FAG^{i}|_{\FAA_i} (a_{\mij}) = \alpha_i$. We have a contradiction.

Then for any available joint action $\delta_\FAG$ of $\FAG$ at $s_0$, if $\sigma_\FAG^{i}|_{\FAA_i} \subseteq \delta_\FAG$, then $\delta_\FAG = \sigma_\FAG^{i}$.

Then $\Fout_{\FAA_i} (s_0, \sigma_\FAG^{i}|_{\FAA_i}) = 
\bigcup \{\Fout_\FAG (s_0, \delta_\FAG) \mid \delta_\FAG \in \FJA_\FAG \text{ and } \sigma_\FAG^{i}|_{\FAA_i} \subseteq \delta_\FAG\} = 
\bigcup \{\Fout_\FAG (s_0, \delta_\FAG) \mid \delta_\FAG \in \Fav_\FAG (s_0) \text{ and } \sigma_\FAG^{i}|_{\FAA_i} \subseteq \delta_\FAG\} = 
\Fout_\FAG (s_0, \sigma_\FAG^{i})$. Note $\Fout^0_\FAG (s_0, \sigma_\FAG^{i}) = \{s_{\mij} \in S \mid j \in \FPI\}$. Also note for any $j \in \FPI$, $\MM, s_{\mij} \Vdash \phi_i$. Then for every $s \in \Fout_{\FAA_i} (s_0, \sigma_\FAG^{i}|_{\FAA_i})$, $\MM,s \Vdash \phi_i$. Then, $\MM,s_0 \Vdash \Fclo{\FAA_i} \phi_i$.

\begin{example}

Consider Example \ref{example:13}.
How does $(\MM,s_0)$ satisfy 
$
\Fclo{\FAA_{-1}} \phi_{-1} \land \Fclo{\FAA_{-2}} \phi_{-2} \land \Fclo{\FAA_{-3}} \phi_{-3}
$, where $\FAA_{-1} = \{a\}$, $\phi_{-1} = p$, $\FAA_{-2} = \{b\}$, $\phi_{-2} = q$, $\FAA_{-3} = \emptyset$, $\phi_{-3} = \top$?

\begin{itemize}

\item 

Note: 

$\alpha_{-1} \in \Fav^0_a (s_0)$ and $\Fout^0_a (s_0, \alpha_{-1}) = \{s_{\moneone}, s_{\monetwo}, s_{\monethree}\}$;

$\MM_{\moneone},s_{\moneone} \Vdash \top \land p \land \neg (p \land q)$;

$\MM_{\monetwo},s_{\monetwo} \Vdash \top \land p \land \neg (\neg p \lor q)$;

$\MM_{\monethree},s_{\monethree} \Vdash \top \land p \land \neg \bot$.

Then, $\MM,s_0 \Vdash \Fclo{a} p$.

\item 

Note: 

$\alpha_{-2} \in \Fav^0_b (s_0)$ and $\Fout^0_b (s_0, \alpha_{-2}) = \{s_{\mtwoone}, s_{\mtwothree}\}$;

$\MM_{\mtwoone},s_{\mtwoone} \Vdash \top \land q \land \neg (p \land q)$;

$\MM_{\mtwothree},s_{\mtwothree} \Vdash \top \land q \land \neg \bot$.

Then, $\MM,s_0 \Vdash \Fclo{b} q$.

\item 

Note $\emptyset \in \Fav^0_\emptyset (s_0)$.
Then, $\MM,s_0 \Vdash \Fclo{\emptyset} \top$.

\end{itemize}
\end{example}

\medskip

\textbf{Step 2c: showing $\neg \Fclo{\FBB_j} \psi_j$ is true at $(\MM,s_0)$ for every $j \in \FPI$.}

\medskip

Let $j \in \FPI$.

Assume $\MM,s_0 \Vdash \Fclo{\FBB_j} \psi_j$. Then at $s_0$, $\FBB_j$ has an available joint action $\sigma_{\FBB_j}$ such that $\sigma_{\FBB_j} \leadto_{(\MM,s_0)} \psi_j$.

Then $\sigma_{\FBB_j} \subseteq \delta_\FAG$ for some available joint action $\delta_\FAG$ of $\FAG$ at $s_0$. Note at $s_0$, there are only two kinds of available joint actions for $\FAG$: $\sigma^{m}_\FAG$ and $\Fspac^{m\text{-}n}_\FAG$.

Assume $\delta_\FAG = \Fspac^{m\text{-}n}_\FAG$ for some $m \in \FNI$ and $n \in \FPI$. Note $\Fout_\FAG (s_0, \Fspac^{m\text{-}n}_\FAG) = S$. Then $\Fout_{\FBB_j} (s_0, \sigma_{\FBB_j})$ $= S$.
Note there is $i' \in \FNI$ such that $\FAA_{i'} \subseteq \FBB_j$. Then $\MM,s_{i'\text{-}j} \Vdash \phi_{\FNI_0} \land \phi_{i'} \land \neg \psi_j$. We have a contradiction. 

Then $\delta_\FAG = \sigma^i_\FAG$ for some $i \in \FNI$, and $\sigma_{\FBB_j} \subseteq \sigma^i_\FAG$.

Assume $\FAA_i \subseteq \FBB_j$. Then $s_{\mij} \in \Fout_\FAG (s_0, \sigma_\FAG^{i}) \subseteq \Fout_{\FBB_j} (s_0, \sigma_{\FBB_j})$. Note $\MM,s_{\mij} \Vdash \phi_{\FNI_0} \land \phi_{i} \land \neg \psi_j$. We have a contradiction.

Assume $\FAA_i \not \subseteq \FBB_j$. Note $\Fspac^{\mij}_\FAG$ is different from $\sigma_\FAG^{i}$ only at $a_{\mij}$, which is not in $\FBB_j$. Then $\Fspac^{\mij}_\FAG|_{\FBB_j} = \sigma_\FAG^{i}|_{\FBB_j}$. Note $\Fout_\FAG (s_0, \Fspac^{\mij}_\FAG) = S$. Then $\Fout_{\FBB_j} (s_0, \sigma_\FAG^{i}|_{\FBB_j}) = S$. Then $\Fout_{\FBB_j} (s_0, \sigma_{\FBB_j}) = S$.
Note there is $i' \in \FNI$ such that $\FAA_{i'} \subseteq \FBB_j$. Then $\MM,s_{i'\text{-}j} \Vdash \phi_{\FNI_0} \land \phi_{i'} \land \neg \psi_j$. We have a contradiction.

\begin{example}

Consider Example \ref{example:13}.
How does $(\MM,s_0)$ satisfy
$
\neg \Fclo{\FBB_{1}} \psi_1 \land \neg \Fclo{\FBB_{2}} \psi_2 \land \neg \Fclo{\FBB_{3}} \psi_3
$, where $\FBB_{1} = \FAG$, $\psi_1 = p \land q$, $\FBB_{2} = \{a\}$, $\psi_2 = \neg p \lor q$, $\FBB_{3} = \FAG$, $\psi_3 = \bot$?

\begin{itemize}

\item 

Note:

$\Fav^0_\FAG (s_0) = \{\sigma_\FAG^{-1}, \sigma_\FAG^{-1}, \sigma_\FAG^{-1}, \lambda^{\mtwotwo}_\FAG\}$;

$\Fout^0_\FAG (s_0, \sigma_\FAG^{-1}) = \{s_{\moneone}, s_{\monetwo}, s_{\monethree}\}$;

$\Fout^0_\FAG (s_0, \sigma_\FAG^{-2}) = \{s_{\mtwoone}, s_{\mtwothree}\}$;

$\Fout^0_\FAG (s_0, \sigma_\FAG^{-3}) = \{s_{\mthreeone}, s_{\mthreetwo}, s_{\mthreethree}\}$;

$\Fout^0_\FAG (s_0, \lambda^{\mtwotwo}_\FAG) = S$;

$\sigma_\FAG^{-1} \not \leadto_{\MM,s_0} p \land q$, as $\MM_{\moneone},s_{\moneone} \Vdash \top \land p \land \neg (p \land q)$;

$\sigma_\FAG^{-2} \not \leadto_{\MM,s_0} p \land q$, as $\MM_{\mtwoone},s_{\mtwoone} \Vdash \top \land q \land \neg (p \land q)$;

$\sigma_\FAG^{-3} \not \leadto_{\MM,s_0} p \land q$, as $\MM_{\mthreeone},s_{\mthreeone} \Vdash \top \land \top \land \neg (p \land q)$;

$\lambda^{\mtwotwo}_\FAG \not \leadto_{\MM,s_0} p \land q$, as $(\MM_{\mthreeone},s_{\mthreeone}) \Vdash \top \land \top \land \neg (p \land q)$.

Then, $\MM,s_0 \not \Vdash \Fclo{\FAG} (p \land q)$.

\item

Note:

$\Fav^0_a (s_0) = \{\alpha_{-1}, \alpha_{-2}, \alpha_{-3}\}$;

$\Fout^0_a (s_0, \alpha_{-1}) = \{s_{\moneone}, s_{\monetwo}, s_{\monethree}\}$;

$\Fout^0_a (s_0, \alpha_{-2}) = S$;

$\Fout^0_a (s_0, \alpha_{-3}) = \{s_{\mthreeone}, s_{\mthreetwo}, s_{\mthreethree}\}$;

$\alpha_{-1} \not \leadto_{\MM,s_0} \neg p \lor q$, as $\MM_{\monetwo},s_{\monetwo} \Vdash \top \land p \land \neg (\neg p \lor q)$;

$\alpha_{-2} \not \leadto_{\MM,s_0} \neg p \lor q$, as $\MM_{\monetwo},s_{\monetwo} \Vdash \top \land p \land \neg (\neg p \lor q)$;

$\alpha_{-3} \not \leadto_{\MM,s_0} \neg p \lor q$, as $\MM_{\mthreetwo},s_{\mthreetwo} \Vdash \top \land \top \land \neg (\neg p \lor q)$.

Then, $\MM,s_0 \not \Vdash \Fclo{a} (\neg p \lor q)$.

\item 

Clearly, $\MM,s_0 \not \Vdash \Fclo{\FAG} \bot$.

\end{itemize}

\end{example}

\end{proof}

\subsection{Upward derivability lemma of standard formulas}

\begin{lemma}
\label{lemma:??}
The following formula and rule are derivable:

\medskip

\begin{tabular}{rl}
Special independence of agents ($\mathtt{A}\text{-}\mathtt{SIA}$): & $(\Fclo{\emptyset} \phi \land \Fclo{\FAA} \psi) \rightarrow \Fclo{\FAA} (\phi \land \psi)$ \vspace{5pt} \\
Monotonicity ($\mathtt{R}\text{-}\mathtt{Mon}$): & $\dfrac{\phi \rightarrow \psi}
{\Fclo{\FAA} \phi \rightarrow \Fclo{\FBB} \psi}$ where $\FAA \subseteq \FBB$
\end{tabular}

\end{lemma}

\begin{proof}
~

First, we show that the formula $\mathtt{A}\text{-}\mathtt{SIA}$ is derivable.
Note $\vdash \phi \rightarrow (\psi \rightarrow (\phi \land \psi))$.
By Rule $\mathtt{R}\text{-}\mathtt{CN}$, $\vdash \Fclo{\emptyset} \phi \rightarrow \Fclo{\emptyset} (\phi \rightarrow (\psi \rightarrow (\phi \land \psi)))$.
By Axiom $\mathtt{A}\text{-}\mathtt{MG}$, $\vdash \Fclo{\emptyset} (\phi \rightarrow (\psi \rightarrow (\phi \land \psi))) \rightarrow (\Fclo{\emptyset} \phi \rightarrow \Fclo{\emptyset} (\psi \rightarrow (\phi \land \psi)))$.
Then $\vdash \Fclo{\emptyset} \phi \rightarrow (\Fclo{\emptyset} \phi \rightarrow \Fclo{\emptyset} (\psi \rightarrow (\phi \land \psi)))$.
Then $\vdash \Fclo{\emptyset} \phi \rightarrow \Fclo{\emptyset} (\psi \rightarrow (\phi \land \psi))$.
By Axiom $\mathtt{A}\text{-}\mathtt{MG}$, $\vdash \Fclo{\emptyset} (\psi \rightarrow (\phi \land \psi)) \rightarrow (\Fclo{\FAA} \psi \rightarrow \Fclo{\FAA} (\phi \land \psi))$.
Then $\vdash \Fclo{\emptyset} \phi \rightarrow (\Fclo{\FAA} \psi \rightarrow \Fclo{\FAA} (\phi \land \psi))$.
Then $\vdash (\Fclo{\emptyset} \phi \land \Fclo{\FAA} \psi) \rightarrow \Fclo{\FAA} (\phi \land \psi)$.

Second, we show the rule $\mathtt{R}\text{-}\mathtt{Mon}$ is derivable.
Assume $\vdash \phi \rightarrow \psi$ and $\FAA \subseteq \FBB$.
By Rule $\mathtt{R}\text{-}\mathtt{CN}$, $\vdash \Fclo{\FAA} \phi \rightarrow \Fclo{\emptyset} (\phi \rightarrow \psi)$.
By Axiom $\mathtt{A}\text{-}\mathtt{MG}$, $\vdash \Fclo{\emptyset} (\phi \rightarrow \psi) \rightarrow (\Fclo{\FAA} \phi \rightarrow \Fclo{\FAA} \psi)$.
Then, $\vdash \Fclo{\FAA} \phi \rightarrow (\Fclo{\FAA} \phi \rightarrow \Fclo{\FAA} \psi)$. Then $\vdash \Fclo{\FAA} \phi \rightarrow \Fclo{\FAA} \psi$.
By Axiom $\mathtt{A}\text{-}\mathtt{MC}$, $\vdash \Fclo{\FAA} \psi \rightarrow \Fclo{\FBB} \psi$. Then $\vdash \Fclo{\FAA} \phi \rightarrow \Fclo{\FBB} \psi$.

\end{proof}

\begin{lemma}[Upward derivability of standard formulas]
\label{lemma:Upward derivability of standard formulas}
Let $\gamma \lor (\FBW_{i \in \FNI} \Fclo{\FAA_i} \phi_i \rightarrow \FBV_{j \in \FPI} \Fclo{\FBB_j} \psi_j)$ be a standard formula.

Assume one of the following two conditions is met:
\begin{enumerate}[label=(\alph*),leftmargin=3.33em]

\item 

$\vdash \gamma$;

\item 

there is $i \in \FNI$ and $j \in \FPI$ such that $\FAA_i \subseteq \FBB_{j}$ and $\vdash (\phi_{\FNI_0} \land \phi_{i}) \rightarrow \psi_{j}$.

\end{enumerate}

Then, $\vdash \gamma \lor (\FBW_{i \in \FNI} \Fclo{\FAA_i} \phi_i \rightarrow \FBV_{j \in \FPI} \Fclo{\FBB_j} \psi_j)$.

\end{lemma}

\begin{proof}
~

Assume (a).
Then, $\vdash \gamma \lor (\FBW_{i \in \FNI} \Fclo{\FAA_i} \phi_i \rightarrow \FBV_{j \in \FPI} \Fclo{\FBB_j} \psi_j)$.

Assume (b).
By the derivable rule $\mathtt{R}\text{-}\mathtt{Mon}$, $\vdash \Fclo{\FAA_i} (\phi_{\FNI_0} \land \phi_{i}) \rightarrow \Fclo{\FBB_j} \psi_{j}$.
By the derivable formula $\mathtt{A}\text{-}\mathtt{SIA}$, $\vdash (\Fclo{\emptyset} \phi_{\FNI_0} \land \Fclo{\FAA_i} \phi_{i}) \rightarrow \Fclo{\FAA_i} (\phi_{\FNI_0} \land \phi_{i})$.
Note $\phi_{\FNI_0} = \FBW \{\phi_i \mid i \in \FNI_0\}$. By repeating using $\mathtt{A}\text{-}\mathtt{SIA}$, $\vdash \FBW_{i \in \FNI_0} \Fclo{\FAA_i} \phi_{i} \rightarrow \Fclo{\emptyset} \phi_{\FNI_0}$.
Then $\vdash (\FBW_{i \in \FNI_0} \Fclo{\FAA_i} \phi_{i} \land \Fclo{\FAA_i} \phi_{i}) \rightarrow \Fclo{\FBB_j} \psi_{j}$.
Clearly, $\vdash \Fclo{\FBB_j} \psi_{j} \rightarrow \FBV_{j \in \FPI} \Fclo{\FBB_j} \psi_j$.
Then $\vdash (\FBW_{i \in \FNI_0} \Fclo{\FAA_i} \phi_{i} \land \Fclo{\FAA_i} \phi_{i}) \rightarrow \FBV_{j \in \FPI} \Fclo{\FBB_j} \psi_j$.
Then $\vdash \FBW_{i \in \FNI} \Fclo{\FAA_i} \phi_i \rightarrow \FBV_{j \in \FPI} \Fclo{\FBB_j} \psi_j$.
Then, $\vdash \gamma \lor (\FBW_{i \in \FNI} \Fclo{\FAA_i} \phi_i \rightarrow \FBV_{j \in \FPI} \Fclo{\FBB_j} \psi_j)$.

\end{proof}

\subsection{Completeness of $\FMCL$ by induction}

\begin{theorem}[Soundness and completeness of $\FMCL$]
The axiomatic system for $\FMCL$ given in Definition \ref{definition:An axiomatic system for MCL} is sound and complete with respect to the set of valid formulas in $\Phi_{\FMCL}$.
\end{theorem}

\begin{proof}
~

The soundness is easy to verify, and we skip its proof.

Let $\phi$ be a formula in $\Phi_{\FMCL}$. Assume $\models \phi$. We want to show $\vdash \phi$. We put an induction on the modal depth of $\phi$.

Assume that the modal depth of $\phi$ is $0$. As $\FMCL$ extends the classical propositional logic, $\vdash \phi$.

Assume that the modal depth of $\phi$ is $n$, greater than $0$.
By Lemma \ref{lemma:normal-form}, the normal form lemma, there is $\phi'$ such that (1) $\vdash_\FMCL \phi \leftrightarrow \phi'$, (2) $\phi$ and $\phi'$ have the same modal depth, and (3) $\phi'$ is in the form of $\chi_0 \land \dots \land \chi_k$, where every $\chi_i$ is a standard formula.

By soundness, $\models \phi \leftrightarrow \phi'$. Then $\models \phi'$. Let $i \leq k$. Then $\models \chi_i$. It suffices to show $\vdash \chi_i$.

Assume that the modal degree of $\chi_i$ is less than $n$. By the inductive hypothesis, $\vdash \chi_i$.

Assume that the modal degree of $\chi_i$ is $n$.
By Lemma \ref{lemma:Downward validity of standard formulas}, the downward validity lemma, one of the following two conditions is met:
\begin{enumerate}[label=(\alph*),leftmargin=3.33em]

\item 

$\vDash \gamma$;

\item 

there is $i \in \FNI$ and $j \in \FPI$ such that $\FAA_i \subseteq \FBB_{j}$ and $\vDash (\phi_{\FNI_0} \land \phi_{i}) \rightarrow \psi_{j}$.

\end{enumerate}

By the inductive hypothesis, one of the following two conditions is met:
\begin{enumerate}[label=(\alph*),leftmargin=3.33em]

\item 

$\vdash \gamma$;

\item 

$\vdash (\phi_{\FNI_0} \land \phi_{i}) \rightarrow \psi_{j}$.

\end{enumerate}

By Lemma \ref{lemma:Upward derivability of standard formulas}, the upward derivability lemma, $\vdash \chi_i$.

\end{proof}

\section{Further work}
\label{section:Further work}

Three kinds of work are worth doing in the future.

In general concurrent game models, we drop the assumptions of seriality, independence of agents, and determinism with concurrent game models. However, we may want to keep some of them when constructing logics for strategic reasoning in some special kinds of situations.
The first kind of work is to show the completeness of the logics determined by general concurrent game models with some of them.
As mentioned above, the only difference between the axiomatic system for $\FCL$ given in Definition \ref{definition:Another axiomatic system for CL} and the axiomatic system for $\FMCL$ given in Definition \ref{definition:An axiomatic system for MCL} is that the latter does not have the axioms $\mathtt{A}\text{-}\mathtt{Ser}$, $\mathtt{A}\text{-}\mathtt{IA}$ and $\mathtt{A}\text{-}\mathtt{Det}$, which are intuitively related to the three assumptions.
A conjecture is that for any logic determined by general concurrent game models with some of them, the extension of the axiomatic system for $\FMCL$ with the related axioms is complete with respect to it.

By the downward validity lemma and the upward derivability lemma, we can show that $\FMCL$ is decidable. What is the computational complexity of its satisfiability problem? This is the second kind of work.

The third kind of work is to study a temporal extension of $\FMCL$, which is similar to $\FATL$ as a temporal extension of $\FCL$.

\subsection*{Acknowledgments}

The authors thank Valentin Goranko, Marek Sergot, Thomas \r{A}gotnes, and Emiliano Lorini for their help with this project.
The authors also thank the audience for a logic seminar at Beijing Normal University and a workshop at Southwest University.

\bibliographystyle{alpha}
\bibliography{Strategy-reasoning,Strategy-reasoning-special}

\end{document}